\documentclass[sigconf]{acmart}
\AtBeginDocument{%
  }

\copyrightyear{2023}
\acmYear{2023}
\setcopyright{acmlicensed}\acmConference[WWW '23]{Proceedings of the ACM Web Conference 2023}{May 1--5, 2023}{Austin, TX, USA}
\acmBooktitle{Proceedings of the ACM Web Conference 2023 (WWW '23), May 1--5, 2023, Austin, TX, USA}
\acmPrice{15.00}
\acmDOI{10.1145/3543507.3583398}
\acmISBN{978-1-4503-9416-1/23/04}
\usepackage{algorithm}
\usepackage[noend]{algpseudocode}
\algblockdefx[Foreach]{Foreach}{EndForeach}[1]{\textbf{for each} #1 \textbf{do}}{\textbf{end for}}
\algnewcommand{\IIf}[1]{\State\algorithmicif\ #1\ \algorithmicthen}
\algnewcommand{\EndIIf}{\unskip\ \algorithmicend\ \algorithmicif}
\usepackage{enumitem}
\usepackage[utf8]{inputenc}
\usepackage{url,xspace}
\usepackage{mathtools}
\usepackage{amsthm}
\usepackage{bbm}
\usepackage{graphicx}
\usepackage[bf]{subfigure}
\usepackage{thmtools,thm-restate}
\hypersetup{
    colorlinks=true,
    linkcolor=blue,
    filecolor=magenta,      
    urlcolor=cyan,
}

\setlist{nolistsep,leftmargin=*}
\DeclareMathOperator*{\argmax}{argmax}

\newtheorem{definition}{\textbf{Definition}}

\newtheorem{proposition}{\textbf{Proposition}}

\newtheorem{problem}{\textbf{Problem}}

\DeclareUnicodeCharacter{B0}{\textdegree}
\DeclareUnicodeCharacter{0301}{\'{e}}
\newcommand{\AnjaliComment}[1]{\textcolor{brown}{[[#1]]}}
\def\sealalgo{\textsc{SeAl}\xspace}
\def\greedyNashalgo{\textsc{GreedyNash}\xspace}
\def\greedyRevalgo{\textsc{GreedyRevenue}\xspace}
\def\fralgo{\textsc{FairRec}\xspace}
\def\FOEIR{\textsc{FOEIR}\xspace}
\def\UnconsGreedyNash{\textsc{UnconsGreedyNash}\xspace}
\def\LPT{\textsc{LPT}\xspace}
\def\TFUalgo{\textsc{TFrom}\xspace}

\def\userConst{\textsc{Re-seller-Constraint}\xspace}
\def\productConst{\textsc{Product-Constraint}\xspace}
\def\NSWCardConst{\textsc{2S-CardNashOpt}\xspace}
\def\NSWOpt{\textsc{NashOptimization}\xspace}

\begin{document}

\setlength{\abovedisplayskip}{3pt}
\setlength{\belowdisplayskip}{3pt}
\setlength{\textfloatsep}{2pt}
\title{Towards Fair Allocation in Social Commerce Platforms}

\author{Anjali Gupta}
\affiliation{
  \institution{Indian Institute of Technology Delhi}
  \country{India}
}

\author{Shreyans J. Nagori}
\affiliation{
  \institution{Indian Institute of Technology Delhi}
  \country{India}
}

\author{Abhijnan Chakraborty}
\affiliation{
  \institution{Indian Institute of Technology Delhi}
  \country{India}
}

\author{Rohit Vaish}
\affiliation{
  \institution{Indian Institute of Technology Delhi}
  \country{India}
}

\author{Sayan Ranu}
\affiliation{
  \institution{Indian Institute of Technology Delhi}
  \country{India}
}

\author{Prajit Prashant Nadkarni}
\affiliation{%
  \institution{Flipkart Internet Pvt Ltd}
  \country{India}
}

\author{Narendra Varma Dasararaju}
\affiliation{%
  \institution{Flipkart Internet Pvt Ltd}
  \country{India}
}
\author{Muthusamy Chelliah}
\affiliation{%
  \institution{Flipkart Internet Pvt Ltd}
  \country{India}
} 

\renewcommand{\shortauthors}{Gupta et al.}

\begin{abstract}
Social commerce platforms are emerging businesses where producers sell products through \emph{re-sellers} who advertise the products to other customers in their social network. Due to the increasing popularity of this business model, thousands of small producers and re-sellers are starting to depend on these platforms for their livelihood; thus, it is important to provide {\it fair earning opportunities} to them. The enormous product space in such platforms prohibits manual search, and motivates the need for recommendation algorithms to effectively allocate product exposure and,  consequently, earning opportunities. In this work, we focus on the fairness of such allocations in social commerce platforms and formulate the problem of assigning products to re-sellers as a fair division problem with indivisible items under \emph{two-sided cardinality constraints}, wherein each product must be given to at least a certain number of re-sellers and each re-seller must get a certain number of products.
\looseness=-1

Our work systematically explores various well-studied benchmarks of fairness---including Nash social welfare, envy-freeness up to one item ($EF1$), and equitability up to one item ($EQ1$)---from both theoretical and experimental perspectives. We find that the existential and computational guarantees of these concepts known from the unconstrained setting do not extend to our constrained model. To address this limitation, we develop a mixed-integer linear program and other scalable heuristics that provide near-optimal approximation of Nash social welfare in simulated and real social commerce datasets. Overall, our work takes the first step towards achieving provable fairness alongside reasonable revenue guarantees on social commerce platforms.
\end{abstract}

\ccsdesc[500]{Information systems~Recommender systems}
\ccsdesc[500]{Theory of computation~Algorithmic mechanism design}
\ccsdesc[500]{Information systems~Electronic commerce}

\keywords{Social Commerce, Fair Division, Two-sided Constrained Allocation, Nash social welfare, Envy-Freeness, EF1, Equitibility, EQ1}

\maketitle

\vspace{-2mm}
\section{Introduction}
Social commerce platforms are online platforms where social networks between users enable commerce. 
These platforms involve two primary stakeholders: {\it producers}, who sell their products on the  platforms, and \textit{re-sellers}, who are the users who help the producers reach out to the customers. In contrast to traditional e-commerce platforms where customers \emph{directly} interact with the platform to purchase products, in social commerce the interaction is \emph{indirect} and is facilitated by re-sellers  
who curate and promote products to other users/customers (see Figure~\ref{fig:SC}). 
Re-sellers leverage their social networks (e.g., WhatsApp, Facebook, Instagram, as well as offline connections) to amplify the potential for sale~\cite{lin2019cross}.  
Typically, the explicit social network being leveraged by the reseller is invisible to the social commerce platform. 

Several companies around the world have successfully adopted social commerce as their business model. 
Prominent examples include Pinduoduo and Taobao in China, Meesho and Shopsy in India, Cafepress and Lockerz in USA, and Shopify in Canada. According to a recent study, social commerce industry is projected to reach USD 7.03 trillion globally by 2030~\cite{strategicmarketresearch}. 
The impressive growth of social commerce stems from its ability to tap into a customer base that e-commerce has been unable to reach. Specifically, a large number of customers are 
still more comfortable 
with physical shopping than electronic marketplaces. This preference is due to 
various factors including 
a lack of digital skills, 
online payment issues, lack of trust over after-sales services, etc.~\cite{forbe}. 
Social commerce 
combines the benefits of e-commerce and physical stores: interacting with the re-sellers provides the customers with the same trust factor (the ``human touch'') and comfort of interacting with physical stores (say, due to the ease of communication in the local dialect~\cite{mint}) while extending the benefits of lower price and a larger item inventory commonly associated with electronic marketplaces.
%
%

By incorporating re-sellers in the business chain, social commerce provides livelihood to millions of people. For example, in India alone, one of the social commerce platforms 
 engages $13$ million re-sellers~\cite{techcruch}. In addition, social commerce, much like e-commerce platforms, provides revenue opportunities to a large number of sellers/producers who sell products on these platforms. Since the number of products in the inventory is often huge, social commerce platforms assist re-sellers by recommending products that are more likely to be purchased during their social sharing~\cite{xu2019relation, gao2020item}. 
 The recommendation engine, therefore, plays a crucial role in determining the revenues earned by both producers and re-sellers. Consequently, it is imperative that the recommendations 
are made in a {\it fair and equitable} manner to ensure a holistic long-term engagement of both parties with the platform. 

While the fairness of recommendation in e-commerce platforms has 
been studied recently~\cite{Multi-stake, fairRec,FairMultiSide,multiFR}, 
such questions
remain 
unexplored 
in the context of social commerce. Furthermore, owing to the nature of the business model and resulting constraints in social commerce platforms, fairness notions from e-commerce do not trivially extend to social commerce. This gives rise to various challenges, as listed below:
\begin{itemize}
\item  \textit{Lack of customer visibility:} Many social commerce platforms do not have access to the social graphs among customers and re-sellers. In such cases, they rely entirely on the re-sellers' ability to promote and sell a product. This scenario is in sharp contrast to e-commerce platforms where the customers are directly visible and the recommendations can therefore be personalized. Thus, the social commerce recommendation, in effect, (virtually) allocates products to the re-sellers for selling. Our work makes the first attempt towards exploring fairness in such setting. 
\item \textit{Re-seller constraints:} Each re-seller may have {\it expertise} only on a subset of products available in the inventory.
Additionally, the number of allocations to each re-seller needs to be capped to ensure the design constraints of the UI, which allows showing only a certain number of products to each re-seller.
\looseness=-1
\item  \textit{Product constraints:} Similarly, there are constraints on the product side as well. The platforms may have tie-ups with different producers demanding that every product get some minimum exposure – meaning that every product must be recommended to at least a few re-sellers.
\end{itemize}
\begin{figure}[t]
\centering  
\subfigure
{\label{fig:Ecommerce} \includegraphics[width=0.18\textwidth]{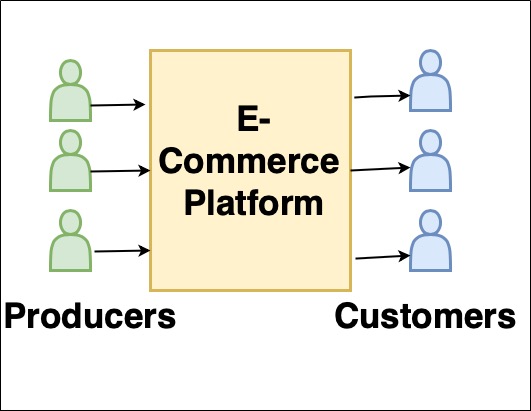}\label{fig:Ecommerce1}}
\hfill
\subfigure
{\label{fig:socialCcommerce} \includegraphics[width=0.22\textwidth]{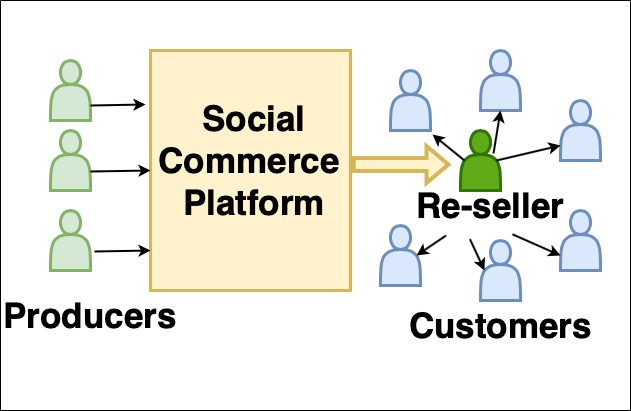}\label{fig:socialCommerce1}}
\caption{Producers connect with the customers \emph{directly} in e-commerce platforms (left), whereas in \emph{social commerce}, re-sellers facilitate this connection (right).
}
\label{fig:SC}
\Description[e-commerce platforms and social commerce]{Producers connect with the customers directly in e-commerce platforms (left), whereas in social commerce, re-sellers facilitate this connection (right).}
\end{figure} 
\noindent
Owing to the challenges outlined above, in this work, we seek to develop a {\it fair product allocation} algorithm for re-sellers which at the same time maximizes \emph{cumulative revenue} of the re-sellers (as well as the social commerce platform) while providing some \emph{minimum exposure guarantee} to each product on the platform. 
Our main contributions are summarized as follows:
\begin{itemize}
\item {\bf Novel problem:} We formulate the problem of assigning products to re-sellers in social commerce platforms as a fair division problem with indivisible items under \emph{two-sided cardinality constraints}. The product-side cardinality constraint ensures that each product must be allocated to at least a certain number of re-sellers. The re-seller cardinality constraint, on the other hand, enforces that each re-seller must get a certain number of products~(\S~\ref{sec:formulation}).
\looseness=-1

\item \textbf{Theoretical analysis:} We explore various well-studied benchmarks of fairness---including \textit{Nash social welfare}, \textit{Envy-Freeness up to one item} (\textit{EF1)}, and \textit{Equitability up to one item} (\textit{EQ1})---in the context of two-sided cardinality constraints. We show that the positive existential and computational guarantees of these concepts in the unconstrained setting (i.e., 
without cardinality constraints) do not extend to our constrained setting~(\S~\ref{FairDivCard}).
\looseness=-1

\item \textbf{Algorithm design:} We develop a mixed-integer linear program (MILP) for the NP-hard problem of optimizing Nash social welfare in our setting. 
While the MILP is reasonable for small-sized simulated data, it turns out to be prohibitively slow for real-world datasets. We overcome this bottleneck by designing considerably faster heuristics that provide near-optimal efficacy~(\S~\ref{sec:algo}).

\item \textbf{Empirical evaluation:} We perform extensive experiments on a large real-world social commerce dataset from Shopsy (owned by Flipkart), one of the largest social commerce platforms in India, to demonstrate the effectiveness and scalability of the proposed methodologies (\S~\ref{sec:experiments}). 

\end{itemize}

\section{Product Allocation in Social Commerce}
\label{sec:formulation}
In this section, we formulate the problem of product allocation in social commerce. First, we define the problem from a purely revenue maximization perspective. Then, we introduce the constraints from both producer and re-seller sides. In the next section, we modify the optimization objective to ensure fair revenue distribution.

\subsection{Revenue Maximization Perspective} 
We begin with the assumption that there exists $m$ re-sellers ($U=\{u_1, u_2,\cdots,u_m\}$) and $n$ products ($P=\{p_1, p_2,\cdots,p_n\}$) in the social commerce platform.
\begin{definition}[Expertise Matrix~($E$)] \textit{
We assume that each re-seller has expertise in a specific set of products which we model via an expertise matrix $E\in\mathbb{R}^{m\times n}$ where $E_{i,j}\in [0,1]$ represents the probability with which re-seller $u_i$ sells a product $p_j$ (after $p_j$ is allocated to $u_i$).}
\end{definition}
In this work, we treat the construction of the expertise matrix as a black box. The typical approach to construct this matrix is through collaborative filtering ~\cite{cf1,cf2} or neural models~\cite{fairCF1, fairCF2, FairCF3}. 
 \looseness=-1

\begin{definition}[Utility~($W$)] \textit{  We denote the expected revenue generated by re-seller $u_i$ for product $p_j$ as Utility $W_{i,j}$.}
\begin{align}
W_{i,j} = E_{i,j} \cdot  \emph{rev}(j)
\end{align}
Here $\emph{rev}(j)$ is the revenue (or profit) generated from selling $p_j$.
\end{definition}
Given a binary re-seller-product \textit{allocation matrix} $A\in [0/1]^{n\times m}$, the expected revenue to the social platform, under the assumption of $W_{i,j}$ independence for all $u_i$ and $p_j$, is:
\begin{align}
\label{eq:rev}
E_{rev}(A) = \sum_{u_i \in U} \sum_{p_j \in P} W_{i,j}\times A_{i,j}
\end{align}
Hereon, we use the notation $A_i= \{\forall p_j\in P,\:A_{i,j}=1\}$ to denote the set of products allocated to re-seller $u_i$. Furthermore, we use $p_j\in A_i$ to denote $A_{i,j}=1$.

A purely revenue maximization objective would therefore reduce to finding the allocation matrix $A_{rev}^*$ that maximizes $E_{rev}$.
\begin{equation}
\arg\max_{A_{rev}^*\in\mathbb{A}} E_{rev}(A_{rev}^*)
\end{equation}
Here, $\mathbb{A}$ denotes the universe of all possible allocation matrices.

\subsection{Two-Sided Constraints}
While maximizing revenue is a reasonable goal, it does not consider some practical business requirements. For example, maximizing revenue may only allocate the most profitable products to the most successful re-sellers. As a result, other products may never get allocated, violating agreements a platform may have with the producers. Similarly, these platforms work within the design constraints of a UI, which typically shows a certain number of products to each re-seller. Without it, some re-sellers may get overwhelmed with product recommendations, while others may fail to get a reasonable number of products, hampering the re-seller experience. 
We introduce the following constraints to avoid these scenarios. 
\begin{enumerate}
\item \textbf{\userConst}: Each re-seller gets allocated at least $L_1$ and at most $L_2$ unique products.
    \begin{align}
     L_1\leq|A_i|\leq L_2, \forall u_i\in U \label{eq:user_constrain}
    \end{align}
\item  \textbf{\productConst}: Each product is allocated to at least $R_1$ and at most $R2$ number of re-sellers.
\begin{align}
        R_1\leq\sum_{i \in \{1,2,...,m\}} \mathbbm{1}\{ p_j \in A_i \} \leq R_2,  \forall p_j \in P \label{eq:product_constrain}
\end{align}
\end{enumerate}
Next, we explore how we can achieve fair revenue distribution under these two constraints.
\section{Fair Allocation under Cardinality Constraints}
\label{FairDivCard}
Prior research in other contexts has shown that focusing only on maximizing revenue can lead to unfairness~\cite{suhr2019two,NEURIPS2019_3070e6ad, FairMultiSide, fairRec, Chakraborty2017FairSF}. Similarly, in social commerce, an algorithm maximizing revenue under the two-sided constraints may allocate all high revenue products to a few re-sellers, while others get only the low revenue ones. If most of the re-sellers do not find enough opportunities to earn well, they may leave a particular platform and go to its competitors. Thus, ensuring that resellers earn equitably can be a differentiating factor to attract more resellers into a platform and retain them in the long run. In this work, we aim to provide a {\it fair product allocation} scheme to re-sellers such that each re-seller is treated fairly in terms of earning opportunity and revenue generation. 

To determine a suitable notion of fairness, we resort to the fair division literature~\cite{FairDivSurvey} and consider two well-studied notions---envy-freeness and equitability---which are based on ``intrapersonal" and ``interpersonal" comparisons, respectively. For indivisible resources (as in our context), exact versions of these notions can not be guaranteed. Hence, one has to define approximations, and two popular approximations are envy-freeness up to one item ($EF1$), and equitability up to one item ($EQ1$). Next, we define $EQ1$ and $EF1$, and prove their non-existence in the presence of \userConst (Eq.\ref{eq:user_constrain}) and \productConst (Eq.\ref{eq:product_constrain}). Then, we discuss another well-studied measure of fairness---Nash social welfare---and establish its suitability in our context where we have to satisfy the two-sided cardinality constraints.

\subsection{Equitability up to One Item (\boldmath$EQ1$)}
\label{subsec:EQ1}
\begin{definition}[Equitability up to one item~($EQ1$)]\textit{An allocation $A = (A_1,\dots,A_n)$ is said to satisfy equitability up to one item ($EQ1$) if for every pair of
re-sellers $u_i, u_k \in U$ such that $A_k\ne\emptyset$, there exists some item $p_j \in A_k$ such that $U_i(A_i) \geq U_k(A_k \setminus {p_j})$. 
} 
\end{definition}
Intuitively, an $EQ1$ allocation ensures that the maximum revenue disparity between any pair of re-sellers is bounded by the revenue of one allotted item. For $EQ1$, positive existence results are known in the literature in the unconstrained setting.
\begin{proposition}
$EQ1$ allocation always exists if an allocation is performed 
without any cardinality constraints~\cite{EQ1,MatroidEQProof}.
\end{proposition} 
However, in the presence of the two-sided constraints, 
an $EQ1$ allocation may not exist.
\begin{restatable}{theorem}{theoremFirst}
\label{cor:eq1fail}
$EQ1$ allocation may not exist in the presence of \userConst (Eq.\ref{eq:user_constrain}) and \productConst (Eq.\ref{eq:product_constrain}).
\end{restatable}
\begin{proof}
Provided in Appendix~\ref{sec:Non_Existence_Proof}.
\end{proof}
\begin{restatable}{theorem}{theoremSecond}
\label{thm:eq1npcomplete}
Determining the existence of an $EQ1$ allocation in the presence of \userConst (Eq.\ref{eq:user_constrain}) and \productConst (Eq.\ref{eq:product_constrain}) is NP-complete.
\end{restatable}
\begin{proof}
Provided in Appendix~\ref{sec:Hardness_Proof}.
\end{proof}
\subsection{Envy-Freeness Up to One item (\boldmath$EF1$)}
\label{subsec:EF1}
 \begin{definition}[Envy-freeness up to one item~($EF1$)]\textit{An allocation $A = (A_1,\dots,A_n)$ is said to satisfy envy-freeness up to one item ($EF1$) if for every pair of re-sellers $u_i, u_k \in U$ such that $A_k\ne\emptyset$, there is some item $p_j \in A_k$ such that $U_i(A_i) \geq U_i(A_k \setminus {p_j})$. Here, $U_i(A_i)= \sum_{p_j \in A_i} W_{i,j}$, where $W_{i,j}$ is the utility of re-seller $u_i$ for the item $p_j$.}
\end{definition}
 $EF1$ ensures that no re-seller envies another re-seller for more than one item. Specifically, if the items allotted to re-seller $u_k$ ($A_k$) are instead allotted to $u_i$, then the maximum increase in the income of $u_i$ is bounded by the revenue of one item in $A_k$. 
\noindent
\looseness=-1
In the unconstrained setting, arbitrarily selecting any permutation of re-sellers and then allocating  products greedily in round robin fashion gives $EF1$ solution~\cite{NSW}. We next show that in the presence of \userConst (Eq.\ref{eq:user_constrain}) and \productConst (Eq.\ref{eq:product_constrain}), the exact permutation of re-sellers matters.
\begin{restatable}{proposition}{propositionThree}
\label{cor:ef1perm}
Some permutation of re-sellers in round robin allocation may not find feasible $EF1$ solution in the presence of \userConst (Eq.\ref{eq:user_constrain}) and \productConst (Eq.\ref{eq:product_constrain}).
\end{restatable}
\begin{proof}
Provided in Appendix~\ref{sec:Non_Existence_Proof}.
\end{proof}

There are pseudopolynomial-time algorithms known for finding an $EF1$ and \textit{Pareto optimal}\footnote{An allocation is Pareto optimal if no other allocation can make some re-seller strictly better off without making some other re-seller strictly worse off.} allocation without cardinality constraints~\cite{BKV} or cardinality constraints at product side~\cite{productCardNSW}. However, their adaptation to our two-sided cardinality-constrained version is not straightforward. \citet{twoSidedFairDivision} studied $EF1$ from both  sides of matching parties (double envy-freeness up to one match) and concluded that two-sided $EF1$ i.e., double $EF1$, does not always exist. ~\citet{fairRec} used round robin allocations to find $EF1$ solution, but without strictly ensuring that product side constraints are satisfied.~\citet{EFTwoSided} considered two-sided $EF1$ in a dynamic setting, but only for symmetric binary valuations. In our case, however, the expected revenue can not be mapped to binary valuations. 

\subsection{Nash social welfare} 
\label{sec:nash}
In the unconstrained setting, the problem of maximizing the Nash social welfare is defined as follows.
\looseness=-1
\begin{problem}[\NSWOpt{}]
\label{prb:NSW}
Given a set of $n$ products $P = \{p_1,\dots,p_n\}$ and $m$ re-sellers $U = \{u_1,\dots,u_m\}$ with utilities \\ $\{W_{i,j}\}_{i \in [n],j \in [m]}$, the goal in \NSWOpt{} is to return an allocation $A^* = (A^*_1,\dots,A^*_m)$ that maximizes the geometric mean (equivalently, the product) of agents' utilities. That is,
\begin{align}
\label{eq:NSW}
A^* \in \argmax_{A = (A_1,\dots,A_n)} \prod_{i \in [m]} \sum_{p_j \in A_i} W_{i,j}
\end{align}
 $ where~A_i\cap A_k=\emptyset~\forall~i,k \in [m]$ and $A_1 \cup \dots \cup A_m = P$.
\end{problem}
Replacing revenue with Nash social welfare in the objective function hits a sweet spot between Bentham’s utilitarian notion of social welfare (maximize the sum of utilities) and the egalitarian notion of Rawls (maximize the minimum utility) in unconstrained cardinality setting\cite{NSW}.

While $EQ1$ and $EF1$ are constraint satisfaction based objectives specifying whether an allocation is permitted or not, \NSWOpt is an optimization problem. Historically, \NSWOpt has mostly been studied in a setting where only one copy of each product is available for allocation. \citet{productCardNSW} proposed a polynomial-time algorithm to approximate the optimal Nash social welfare (\NSWOpt) up to a factor of ${e}^{1/e}$ considering the product side cardinality; however, there is no other constraints (on the  re-seller side cardinality). We focus on optimizing \NSWOpt in two-sided cardinality setting in the presence of \userConst (Eq.\ref{eq:user_constrain}) and \productConst (Eq.\ref{eq:product_constrain}). 

\begin{problem}
\label{prb:NSW_card}
\textbf{Two-sided cardinality constrained Nash Optimization} (\textbf{\NSWCardConst}): Given the set of $m$ re-sellers ($U = \{u_1,\cdots,u_{m}\}$), $n$ products ($P=\{p_1, \cdots,p_{n}\}$), and the parameters $L_1,\;L_2,\;R_1,\;R2$ of the \userConst and \productConst, 
identify the allocation matrix $\mathcal{A}$ maximizing:
\begin{align}
\mathcal{A}=\arg\max_{A} \prod_{u_i \in U} \sum_{p_j \in A_i} W_{i,j} \label{eq:opt_nash2}
\end{align}
subject to satisfying \userConst and \productConst.
\end{problem}

This introduces non-trivial challenges and highlights why obtaining fairness in social commerce is a challenging problem. To rigorously establish this, we next discuss the known results for \NSWOpt in the literature, their validity under two-sided constraints, and why \NSWOpt is the objective function of choice for the proposed problem.

It is known that allocations maximizing Nash social welfare simultaneously satisfy $EF1$ and Pareto optimality---a well-studied criterion of economic efficiency in an unconstrained setting~\cite{NSW}. However, the same result may not hold in the presence of \userConst (Eq.\ref{eq:user_constrain}) and \productConst(Eq.\ref{eq:product_constrain}).
%
%
\begin{restatable}{theorem}{theoremThree}
Optimising Nash social welfare (Eq.~\ref{eq:opt_nash2} ) may not give EF1 solution in the presence of \userConst (Eq.\ref{eq:user_constrain}) and \productConst (Eq.\ref{eq:product_constrain}).
\end{restatable}
\begin{proof}
Provided in Appendix~\ref{sec:Non_Existence_Proof}.
\end{proof}
Next, we prove that two-sided cardinality constrained \NSWOpt is NP-hard, which makes the proposed problem NP-hard.
\vspace{-2mm}
\begin{restatable}{theorem}{theoremFour}
Two-sided cardinality constrained Nash Optimization (\textbf{\NSWCardConst}) problem is NP-hard.
\end{restatable}
\begin{proof}
Provided in Appendix~\ref{sec:Hardness_Proof}.
\looseness=-1
\end{proof}

\subsection{Suitability of Optimizing Nash social welfare for Social Commerce} 
$EQ1$ allocation under two-sided cardinality constraints may not exist (Theorem~\ref{cor:eq1fail}). Furthermore, determining the existence itself is NP-complete (Theorem.~\ref{thm:eq1npcomplete}). Similarly, double-sided $EF1$ may not exist either~\cite{twoSidedFairDivision}. The feasibility of $EF1$ under two-sided cardinality constraints depends upon the selection of re-seller permutation in round-robin allocation (proposition~\ref{cor:ef1perm}). In contrast, although \NSWCardConst is NP-hard, 
we can attempt to design heuristics and identify allocations with high 
Nash social welfare. This motivates us to choose Nash social welfare as the fairness measure of choice. 

Note that by maximizing the (mathematical) product of the individual revenues obtained by each re-seller, we value allocation schemes higher where none of the individual re-seller revenues is low. This ensures re-seller-side fairness. Fairness on the exposure of products is obtained through \productConst.
As we show later, the proposed heuristic for \NSWCardConst provides near-optimal performance empirically. In the next section, we design MILP and our heuristics for \NSWCardConst. 
\looseness=-1

\section{Proposed Algorithms}
\label{sec:algo}
Next, we describe the proposed algorithms to optimize Nash social welfare under two-sided cardinality constraints. Our contributions towards this goal include a mixed integer linear program~(\S~\ref{sec:ILP_Nash}) which is an adaptation of an existing MILP from the unconstrained setting~~\cite{NSW}, and two iterative greedy heuristics~(\S~\ref{subsec:Greedy}).
%
\subsection{Mixed Integer Linear Programming (MILP)}
\label{sec:ILP_Nash}
Since maximizing Nash social welfare in the unconstrained fair division problem is APX-hard~\cite{L17apx}, this motivates the development of integer linear programs and heuristics 
to work on small instances. To this end,~\citet{NSW} proposed a mixed integer linear program (MILP) that was based on the idea of lower bounding the log function by a piecewise linear function.\footnote{Notice that maximizing the Nash welfare product is equivalent to maximizing the sum of log of the utilities.} The latter was designed to be exactly equal to the log function on integral points, and thus a piecewise linear approximation suffices. We adapt this MILP to our setup by incorporating two-sided cardinality constraints and call the adapted program \textsc{NashMax}.

\looseness=-1

\subsubsection{\textsc{NashMax}}
\label{sec:NashMax}
We scale all $W_{i,j}$ (utility) between $1$ to $1000$ integers and let those scaled integer values are $v_{i,j}, ~\forall i$ and $j$ such that $\sum_{j}v_{i,j} \leq 1000, \forall i$. Let $\gamma_{i}$ be a continuous variable denoting the log of the utility to player $i$, and bound it using a set of linear constraints such that the tightest bound at every integral point $k$ is exactly $\log(k)$.
\begin{align}
\text{Maximize } \sum_{i \in U}{\gamma_{i}}
\end{align}
Subject to
\begin{enumerate}
 \item Log approximation: The RHS in the constraint below is a lower bound on the log function everywhere and is tight on all points where the utility is integral.
    \begin{align*}
    \forall i,~  \gamma_{i} \leq \log k +\left[\log (k+1) - \log k\right] \times \left[\sum_{j \in P}{x_{i,j}\times v_{i,j} }-k\right], \nonumber \\ \forall i \in U, k \in \{1,3,\ldots,999\} 
    \end{align*}
\item Checking whether item $j$ is assigned to re-seller $i$
    \begin{align*}
    x_{i,j} \in \{0,1\}
    \end{align*}
\item Each item is allocated to at least $R_1$ re-sellers and at most $R_2$
    \begin{align*}
    \forall j,~ R_1 \leq \sum_{i}{x_{i,j}}\leq R_2
    \end{align*}
\item Each re-seller is allocated at least $L_1$ and at most $L_2$ items
    \begin{align*}
    \forall i,~ L_1 \leq \sum_{j}{x_{i,j}} \leq L_2
    \end{align*}
\end{enumerate}

\begin{proposition}
The allocation returned by \textsc{NashMax} has the highest Nash social welfare among all cardinality-constrained allocations for the given instance.
\label{prop:MILP}
\end{proposition}

We note that the proof of correctness of Proposition~\ref{prop:MILP} is similar to that of the unconstrained version as discussed in \cite{NSW} and is therefore omitted. 
Unfortunately, the MILP described above does not scale to large datasets. This motivates us to develop greedy heuristics for maximizing Nash welfare.

\subsection{Iterative Greedy Heuristics}
\label{subsec:Greedy}
Here we propose two heuristics based approaches to optimize Nash under two-sided cardinality constraints. 
The first approach is based on pure greedy allocation, where we iteratively allocate each product to a re-seller which maximizes the Nash. The second approach is round wise allocation, where the poorest re-seller is preferred to pick up its most preferred product in each round.

\subsubsection{\greedyNashalgo}
\begin{algorithm}[t!]
    \caption{\greedyNashalgo : Greedily Maximize Nash under Two-Sided Cardinality Constraints}\label{algo:algoIGNSW2}
    {\scriptsize
    \begin{flushleft}
        \textbf{Input:} $m$ re-sellers, 
        $n$ products, 
        utilities $W_{ij},~i\leq~m,~j\leq~n$, and 
        parameters $L_1$ and $L_2$ (The lower and upper bound on recommendation list for re-sellers), $R_1$ and $R_2$ (Minimum and maximum number of re-sellers allocated to each product respectively).\\
        \textbf{Output:} Assignments of re-sellers $A = \{A_1,A_2,....,A_m\}$.\\
    \end{flushleft}
    \begin{algorithmic}[1]
    \Statex --------- INITIALIZATION-------------------
    \For {$i = 1 \:to\: m$}: 
    \State set $A_i = \{\}$, (Empty allocations for re-seller $u_i$)
    \EndFor
    \Statex
    \Statex --------- FIRST ALLOCATION-------------------
    \For {$i = 1 \:to\: m$}:
    \State Assign re-seller $u_i$'s most preferred unique product, $p_j$ to $A_i$ subject to cardinality constraints
    \EndFor
    \Statex
    \Statex --------- GREEDY ALLOCATION-------------------
    \For {$j = 1 \:to\: n$}: 
        \Repeat
            \State Find re-seller $u_i$ such that, $p_j \notin A_{i}$, and allocation of product $p_j$ to re-seller $u_i$ leads to 
            \Statex \hspace{0.7cm}maximum increase in Nash subject to cardinality constraints(with lower bounds). Assign product $p_j$ to $A_i$
        \Until {product $p_j$ is allocated to $R_1$ users \textbf{OR} no such suitable user found}
    \EndFor
    \Statex
    \Statex  --------- FULL-FILL RECOMMENDATION LIST  ---------
    \State Some users may not be allocated to $L_1$ products yet. Allocate their most preferred product until their allocation list contains $L_1$ unique products subject to the most preferred product should not be already allocated to $R_2$ users.
    \Statex
    \Statex --------- GREEDY REPLACEMENT ---------
    \For{ each products $p_j$ where number of re-sellers allocated to it is less than $R_1$}:
         \Repeat
            \State Find re-seller $u_i$ such that, $p_l \in A_{i}$ with number of re-sellers allocated to $p_l > R_1$,  
            \Statex \hspace{0.7cm}$p_j \notin A_{i}$, and replacement of $p_l$ with $p_j$ result in a minimum decrease in her/his utility
            \State Replace $p_l$ with $p_j$ in $A_i$
        \Until{Product $p_j$ is allocated to $R_1$ users}
    \EndFor
    \Statex
    \Statex --------- ASSIGN GREEDILY TILL UPPER-BOUNDS SATISFIED ---------
    \For {$j = 1 \:to\: n$}: 
        \Repeat
            \State Find re-seller $u_i$ such that, $p_j \notin A_{i}$, and allocation of product $p_j$ to re-seller $u_i$ leads to 
            \Statex \hspace{0.7cm}maximum increase in Nash subject to cardinality constraints(with upper bounds). Assign product $p_j$ to $A_i$
        \Until {product $p_j$ is allocated to $R_2$ users \textbf{OR} no such suitable user found}
    \EndFor
    \end{algorithmic}}
\end{algorithm}

We now describe an iterative greedy approach that we call \greedyNashalgo (see Algorithm~\ref{algo:algoIGNSW2}). The underlying idea is to allocate each product to myopically achieve the maximum increase in Nash social welfare. 
Specifically, starting from an empty allocation (lines 1-2), we first allocate each re-seller exactly one product by allowing the re-sellers to pick their most preferred product in a round-robin manner subject to maintaining cardinality constraints (lines 3-4). This step ensures that each re-seller acquires a non-zero utility (thus, in turn, making the Nash product non-zero) and makes the remaining decisions of the algorithm well-defined. 

Next, we iterate over the set of products, and, at each step, assign each product in a way that maximizes the marginal increase in Nash Welfare subject to the cardinality constraints (lines 5-8). 
This process continues until either all $R_1$ copies of a product have been allocated or no suitable user is found (line 8). 
If the cardinality constraint for one or more users is not met, we continue allocating them their most-preferred products 
(line 9). Similarly, products not achieving their minimum cardinality are greedily assigned to re-sellers 
(lines 10-14). Once all products achieve their minimum cardinality, we can greedily allocate $R_2$ copies of products subject to user-side cardinality constraints (lines 15-18). 
\looseness=-1

\begin{algorithm}[t]
    \caption{\sealalgo : Sequential Egalitarian Algorithm under Two-Sided Cardinality Constraints}\label{algo:algoSEAL}
    {\scriptsize
    \begin{flushleft}
        \textbf{Input:} $m$ re-sellers, 
        $n$ products, 
        utilities $W_{ij},~i\leq~m,~j\leq~n$, and 
        parameters $L_1$ and $L_2$ (The lower and upper bound on recommendation list for re-sellers), $R_1$ and $R_2$ (Minimum and maximum number of re-sellers allocated to each product, respectively).\\
        \textbf{Output:} Assignments of re-sellers $A = \{A_1,A_2,....,A_m\}$.\\
    \end{flushleft}
    \begin{algorithmic}[1]
    \Statex --------- INITIALIZATION-------------------
    \State $PQ = []$ (Empty ascending order priority queue)
    \For {$i = 1 \:to\: m$}: 
    \State set $A_i = \{\}$, (Empty allocations for re-seller $u_i$)
    \State Insert pair $<u_i,0>$ i.e. $<$user $u_i$, $u_i$'s utility$>$ to $PQ$
    \EndFor
    \Statex
    \Statex --------- SEQUENTIAL ITERATIVE GREEDY ALLOCATION WITH LOWER BOUND--------- 
    \For {$k = 1 \:to\: L_1$}: \label{state:roundStart222}
     \While {$PQ$ is empty}
     \State Take out a reseller $u_i$ with minimum utility from $PQ$\label{state:roundAStart222}
     \State Assign re-seller $u_i$ to his/her most preferred product, $p_j$ such that $p_j \notin A_{i}$ and the
     \Statex \hspace{0.7cm} number of re-sellers allocated to $p_j$ is less than $R_1$
     \State If no such $p_j$ is found then assign re-seller $u_i$ to his/her most preferred product, $p_j$ 
     \Statex \hspace{0.7cm} such that $p_j \notin A_{i}$ and the number of re-sellers allocated to $p_j$ is less than $R_2$
     \EndWhile \label{state:roundAEnd222}
     \For {$i=1\:to\: m$}:
     \State Insert pair $<u_i$,$u_i$'s utility$>$ to $PQ$
     \EndFor
    \EndFor \label{state:roundEnd222}
    \Statex
    \Statex --------- GREEDY REPLACEMENT ---------
    \For{ each products $p_j$ where number of re-sellers allocated to it is less than $R_1$}:
         \Repeat
            \State Find re-seller $u_i$ such that, $p_l \in A_{i}$ with number of re-sellers allocated to $p_l > R_1$,  
            \Statex \hspace{0.7cm}$p_j \notin A_{i}$, and replacement of $p_l$ with $p_j$ result in a minimum decrease in her/his utility
            \State Replace $p_l$ with $p_j$ in $A_i$
        \Until{Product $p_j$ is allocated to $R_1$ users}
    \EndFor

    \Statex
    \Statex --------- SEQUENTIAL ITERATIVE GREEDY ALLOCATION WITH UPPER BOUND--------- 
    \For {$k = L_1 \:to\: L_2$}: \label{state:roundStart222}
     \While {$PQ$ is empty}
     \State Take out a reseller $u_i$ with minimum utility from $PQ$\label{state:roundAStart222}
     \State Assign re-seller $u_i$ to his/her most preferred product, $p_j$ such that $p_j \notin A_{i}$ and the
     \Statex \hspace{0.7cm} number of re-sellers allocated to $p_j$ is less than $R_2$
     \EndWhile \label{state:roundAEnd222}
     \For {$i=1\:to\: m$}:
     \State Insert pair $<u_i$,$u_i$'s utility$>$ to $PQ$
     \EndFor
    \EndFor \label{state:roundEnd222}
    \end{algorithmic}}
\end{algorithm}
\subsubsection{\sealalgo}
It is a sequential egalitarian algorithm (Algorithm~\ref{algo:algoSEAL}). 
Initially, allocations are empty; hence the utility of all re-sellers is $0$. We perform allocations in $L_1$ round. Each round sequentially allocates the poorest re-seller (with the least utility) their most preferred product subject to lower bound cardinality constraints. 
At the end of $L_1$ rounds, if some products are not allocated to $R_1$ re-sellers, perform greedy product replacement as mentioned in steps 12-16. Once every product achieves their minimum cardinality, we can allocate $R_2$ copies of products subject to user-side cardinality constraints (steps 17-22). 
\looseness=-1
\par
Between \greedyNashalgo and \sealalgo, the outperformance 
depends on the utility values of the re-sellers (Appendix~\ref{sec:Examples} provides some relevant Examples). Computational complexity of both algorithms are discussed in Appendix~\ref{sec:complexity}.
Note that in both cases, a feasible solution may not exist under all possible cardinality constraints. However, the constraints can be set to ranges that guarantee feasibility (\S~\ref{sec:experiments}).
\vspace{3mm}
\section{Datasets}
\label{sec:data}
In this section, we introduce the datasets used to benchmark the proposed methodologies.

\vspace{-2mm}
\subsection{Real Social Commerce Dataset}

The dataset from Shopsy comprises of one month of click and purchases history, followed by the next six months' purchase history. In addition, it includes the product catalog, which details the price, revenue (profit), and category of each product. Table~\ref{tab:realData} presents the statistics of this dataset. It is worth noting that this is possibly the largest real dataset being studied in the fair allocation literature.

Since the dataset contains products of various price-ranges and categories, it is not fair to compare the revenues generated by re-sellers across different categories with vastly different expertise. As an example, re-sellers specializing in expensive electronic items (Ex: TV, laptop, etc.) have significantly different expectations in terms of product allocations and revenue contribution to platform when compared to those who specialize in garments. Hence, we partition re-sellers into homogeneous clusters by performing $k$-means clustering on their expertise vectors. The expertise vectors are learned using collaborative filtering~\cite{cf1}. 
The value of $k$ is set to $53$ since inter-cluster distance computed through elbow-method~\cite{elbow} stabilizes at this value. The problem of \NSWCardConst is solved on each cluster separately, and the results are reported in \S~\ref{sec:experiments}. 
The distributions of the cluster sizes and the average price of items sold by re-sellers in each cluster are shown in Fig.~\ref{fig:clusterdist}. As can be seen, there is a significant disparity in the price of the items sold by each cluster, which necessitates this segregation step of the raw data through clustering.

\begin{table}[t]
\centering
\scalebox{0.9} {
\begin{tabular}{cccc} 
    \toprule
    {\bf \#Re-sellers} & {\bf \#Products} & {\bf \#Purchases}& {\bf \#Product categories}\\  
    $531933$ & $13178$ & $951664$ & $703$\\
 \bottomrule
\end{tabular}}
\caption{Summary of the real social commerce dataset.}
\label{tab:realData}
\vspace{-0.10in}
\end{table}

\begin{figure}[t]
\vspace{-0.30in}
\centering  
\subfigure
{\label{fig:resellerDist} \includegraphics[width=0.23\textwidth]{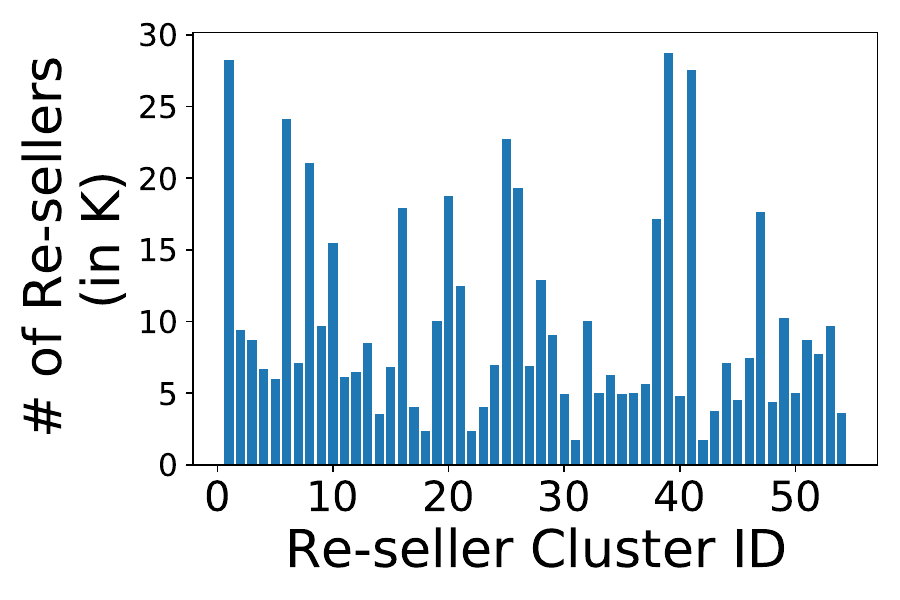}}
\hfill
\subfigure
{\label{fig:priceDist} \includegraphics[width=0.23\textwidth]{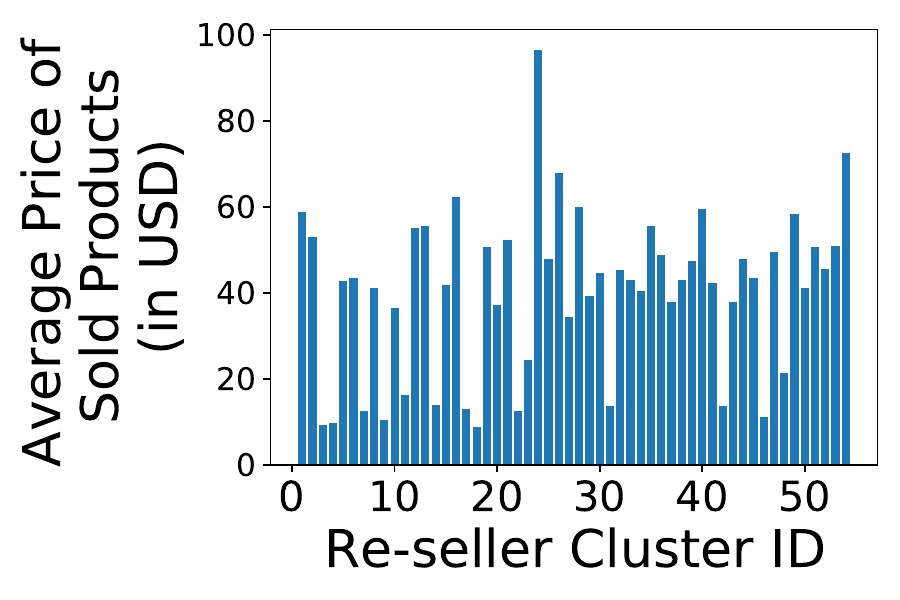}}
\caption{Distributions of the re-sellers (left) and the average price of products sold by them in each cluster (right).
}
 \label{fig:clusterdist}
 \Description[Cluster distributions]{The distributions of the re-sellers (left) and the average price of products sold by re-sellers in each cluster (right).}
\end{figure} 

\subsection{Synthetic Dataset}
Since MILP (Recall \S~\ref{sec:ILP_Nash}) does not scale to large datasets, we create small, synthetic datasets so that the proposed heuristic can be compared to the optimal allocation.
More details about the synthetic dataset generation is provided in Appendix ~\ref{sec:SyntheticDataGeneration}.
\section{Experimental results}
\label{sec:experiments}
In this section, we benchmark the proposed methodologies on the datasets described in \S~\ref{sec:data} and establish that:
\begin{itemize}
\item {\bf Approximation quality:} \sealalgo achieves near optimal performance with minimal dip in revenue.
\item {\bf Efficacy: }\sealalgo provides significantly higher fairness while also ensuring minimal reduction in the revenue when compared to purely maximizing revenue.
\end{itemize}

\noindent
\DeclarePairedDelimiter\ceil{\lceil}{\rceil}
\DeclarePairedDelimiter\floor{\lfloor}{\rfloor}
The implementations of all algorithms used in this study are in Python 3.0. Our experiments are performed on a machine with Intel(R) Xeon(R) CPU @ 2.10GHz with 252GB RAM on Ubuntu 18.04.3 LTS. Our codebase is available at \url{https://github.com/idea-iitd/FairAllocInSocialCom.git}

\begin{figure*}[t]
\centering  
\subfigure[{Re-seller's Fairness}]{\includegraphics[width=1.7in]{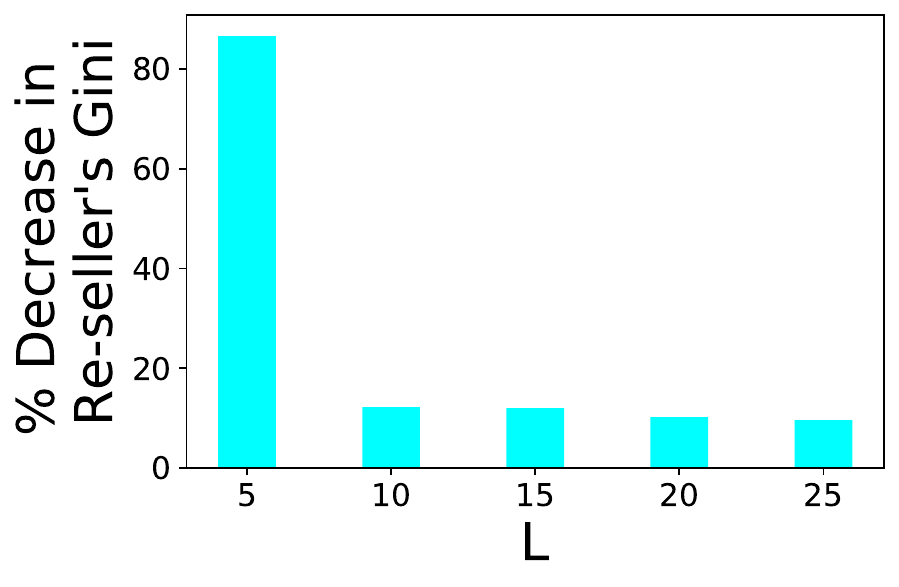}}\hfill
\subfigure[{Producer's Fairness}]{\includegraphics[width=1.7in]{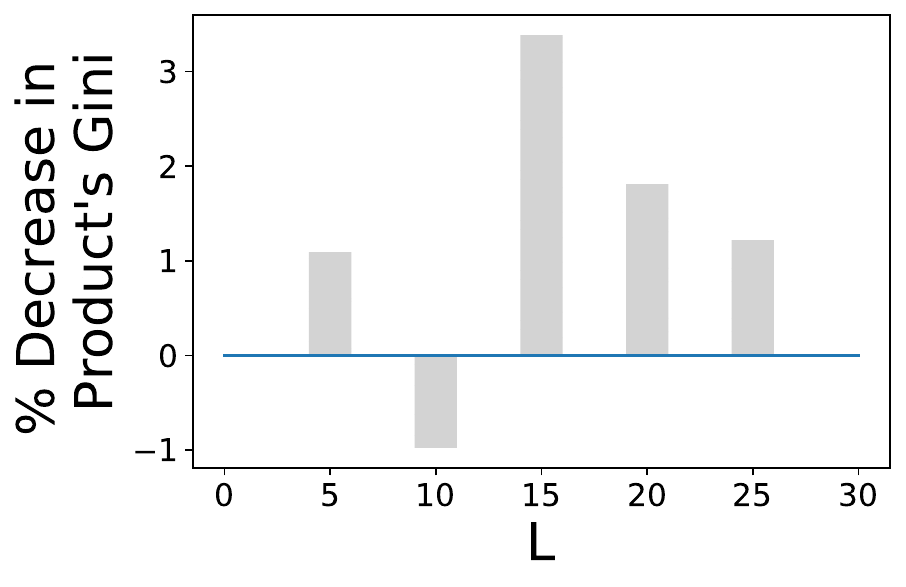}}\hfill
\subfigure[{ Revenue }]{\includegraphics[width=1.7in]{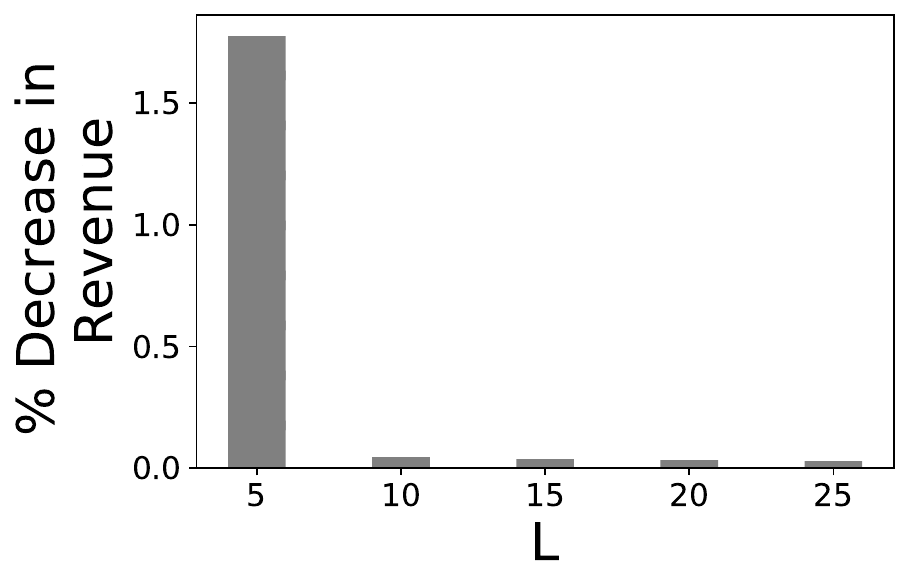}}\hfill
\subfigure[{Re-seller's Fairness}]{\includegraphics[width=1.7in]{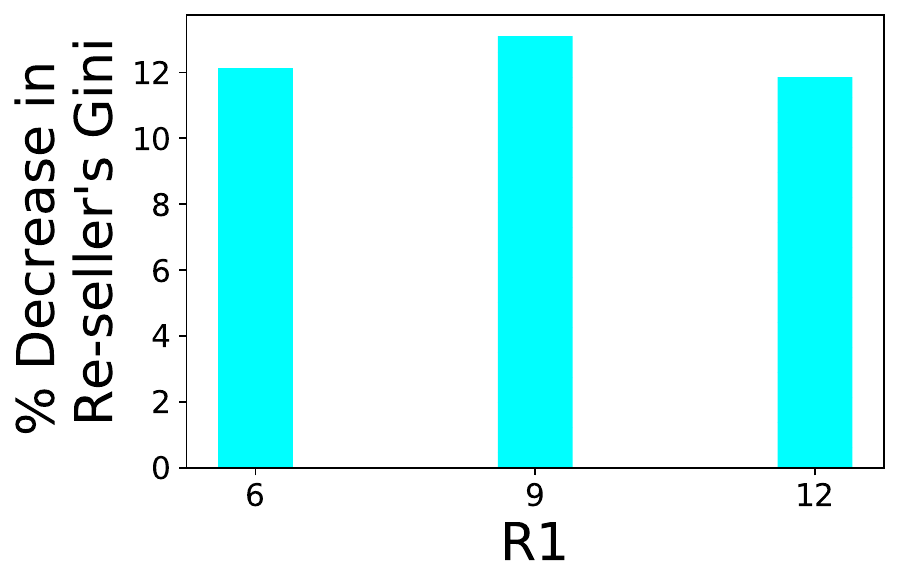}}\\
\subfigure[{Producer's Fairness}]{\includegraphics[width=1.7in]{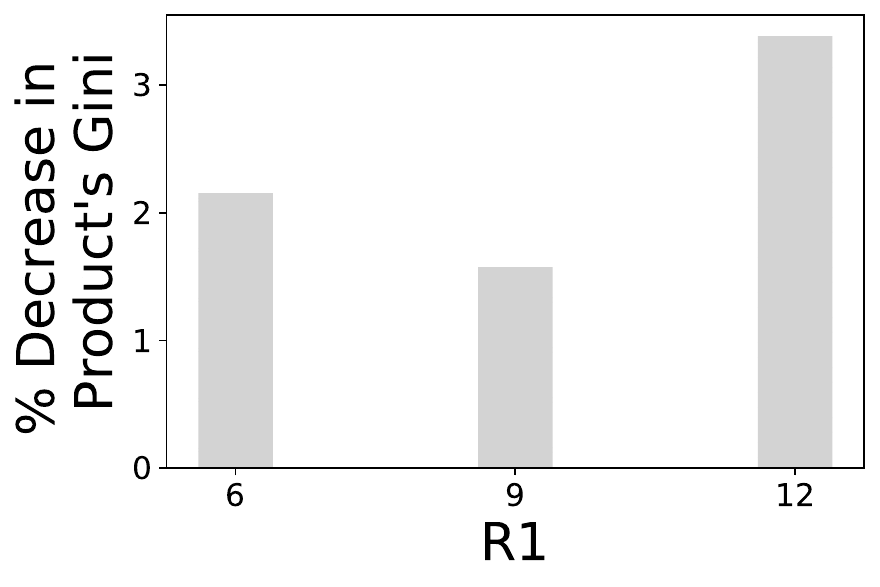}}\hfill
\subfigure[{ Revenue }]{\includegraphics[width=1.7in]{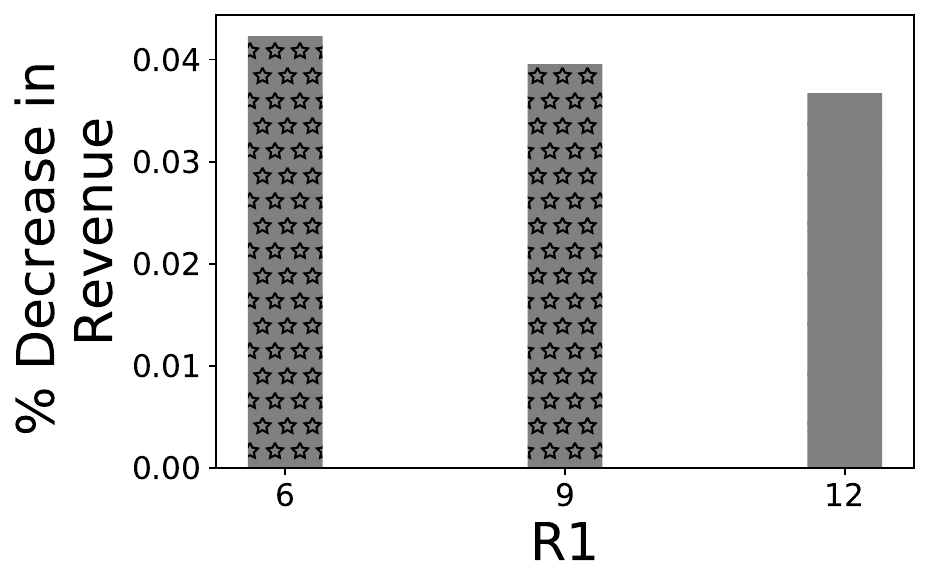}}
\subfigure[]
{\label{fig:resellertime} \includegraphics[width=1.7in]{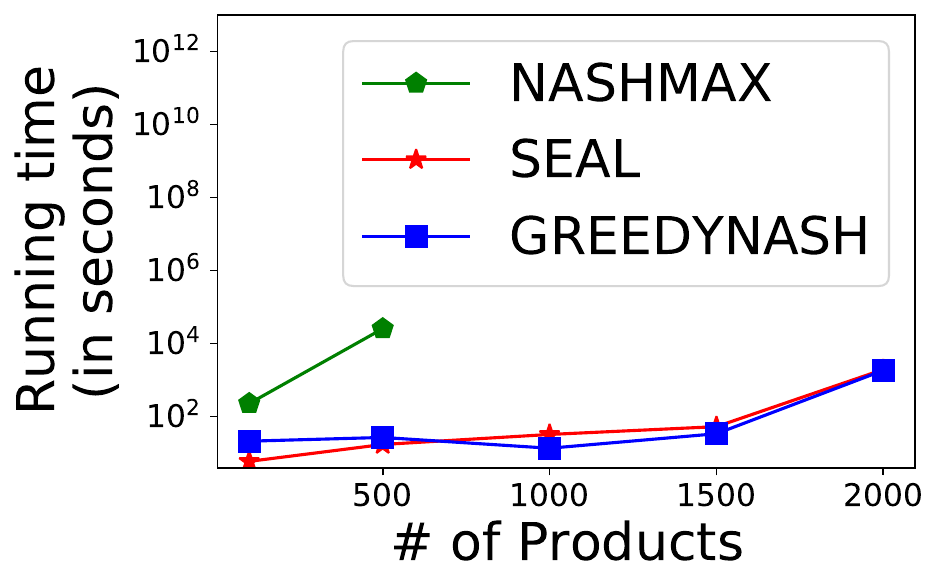}}
\hfill
\subfigure[]
{\label{fig:producttime} \includegraphics[width=1.7in]{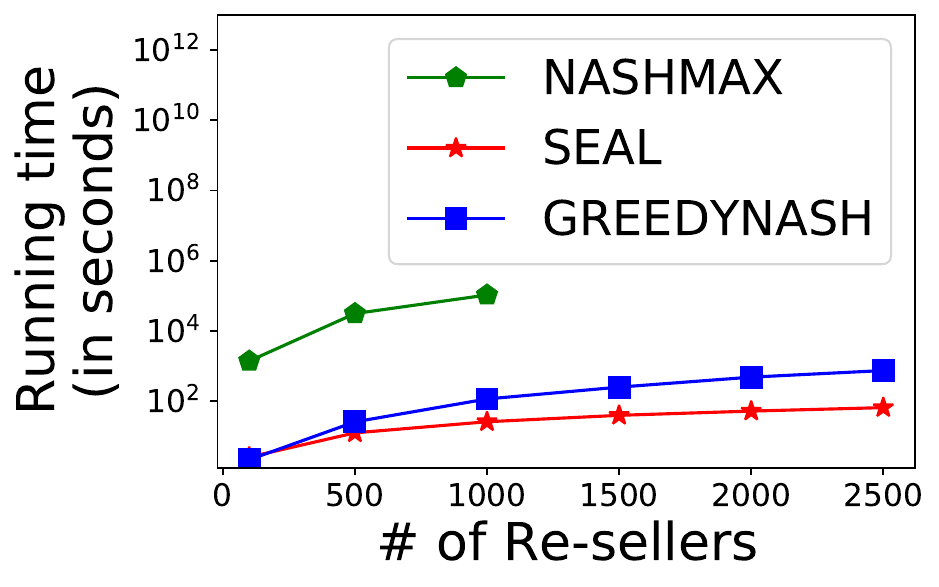}}
\caption{The bar plot shows (a) an increase in fairness among re-sellers, (b) a decrease/increase in fairness among products, and (c) a decrease in revenue on comparing \textsc{NashFair} to \textsc{RevMax} w.r.t. variation in \boldmath{$L$}.  Recall that $L_1$ and $L_2$ are set as $L-\epsilon$ and $L+\epsilon$ respectively where $\epsilon=3$. (d-f) represent the same metrics but, w.r.t. \boldmath{$R1$}. (g-h) Scalability: Running time (in seconds) comparison of \sealalgo, \greedyNashalgo, and \textsc{NashMax} on synthetic dataset w.r.t. increase in the number of (g) products and (h) re-sellers. }
\label{fig:sytheticDataResult}
\Description[Empirical Evaluation]{The bar plot shows (a) an increase in fairness among re-sellers, (b) a decrease/increase in fairness among products, and (c) a decrease in revenue on comparing \textsc{NashFair} to \textsc{RevMax} w.r.t. variation in \boldmath{$L$}.  Recall that $L_1$ and $L_2$ are set as $L-\epsilon$ and $L+\epsilon$ respectively where $\epsilon=3$. (d-f) represent the same metrics but, w.r.t. \boldmath{$R1$}. (g-h) Scalability: Running time (in seconds) comparison of \sealalgo, \greedyNashalgo, and \textsc{NashMax} on synthetic dataset w.r.t. increase in the number of (g) products and (h) re-sellers.}
\end{figure*}
\vspace{-0.05in}
\subsection{Baselines}
To contextualize the performance of the proposed algorithms, we compare its performance to the following baselines:
\begin{itemize}
    
    \item \textsc{RevMax}: We use Integer Linear Programming (ILP) to maximize the overall revenue (Eq.~\ref{eq:rev}) by enforcing both \userConst and \productConst.

    \item \greedyRevalgo: We use the iterative greedy approach to maximize revenue in the presence of \userConst (Eq.\ref{eq:user_constrain}) and \productConst (Eq.\ref{eq:product_constrain}). Specifically, it proceeds in an iterative manner and in each iteration we select the $\langle re-seller, product\rangle$ tuple that maximizes revenue without violating either \userConst or \productConst. While \textsc{RevMax} is guaranteed to provide a superior revenue, \greedyRevalgo is significantly faster and hence scales to large datasets.
    \looseness=-1
    
    
    \item \fralgo~\cite{fairRec} and \textsc{TFrom}~\cite{TFROM} are two-sided fairness-aware recommendation algorithms. \fralgo ensures EF1 allocations for customers, and aims to achieve product side constraints. \textsc{TFrom} considers the position bias in recommendations.
    
    \item \textsc{FOEIR}~\cite{FOEIR} is fair-ranking algorithm that deal with advantaged and disadvantaged social groups' bias and positional bias. In our adaptation, we randomly divided products into these two groups, assumed a uniform attention model and applied the \FOEIR algorithm maintaining the cardinality constraints.
    
\end{itemize}
In addition to the baselines, this work contributes three algorithms, namely \textsc{NashMax}, which optimizes Nash through MILP (\S~\ref{sec:ILP_Nash}), \textsc{GreedyNash} (Alg.~\ref{algo:algoIGNSW2}) and \textsc{SeAl} (Alg.~\ref{algo:algoSEAL}). Also note that we do not consider revenue maximization without constraints, since such allocation results in the highest priced product being allocated to all re-sellers, which is impractical.
\subsection{Parameters}
\textbf{Real dataset: }We set $L_1=L-\epsilon$ and $L_2=L+\epsilon$ for $\epsilon= 0$ and $L =15$ in \userConst, and $R_1=\floor{\frac{L_1+L_2}{2} \times \frac{m}{n}}$, $R_2 = m$ in \productConst. Here, $n$ and $m$ represent the number of products and re-sellers respectively.
\textbf{Synthetic dataset: }
 We vary $L = [5, 10, 15, 20, 25]$, where $L_1=L-\epsilon$ and $L_2=L+\epsilon$ for $\epsilon=3$. We set the minimum guaranteed copy for each product to $R_1$, i.e., $R_1 = \floor{\alpha \times L_1 \times m /n}$ for $\alpha = [0.5,0.75, 1]$ and $R_2=[m, 2 \times R_1]$. 

\subsection{Metrics}
We quantify performance using the following metrics.
\begin{itemize}
\item \textbf{Revenue:} The total expected revenue from all products (Eq.~\ref{eq:rev}). 
\item \textbf{Gini coefficient:} The Gini coefficient is the ratio of the area that lies between the line of equality and the Lorenz curve over the total area under the line of equality ~\cite{gini_coef}.
 Mathematically,
 \begin{equation}
 \label{eq:gini}
Gini=\frac{\sum_{i=1}^{m} \sum_{j=1}^{m}\left|U_{i}-U_{j}\right|}{2 m \sum_{j=1}^{m} U_{j}}
\end{equation}
where $U_i$ is the expected revenue (utility) of re-seller $i$ and $n$ is the number of re-sellers for computing fairness for re-sellers. For products, $i$ is product and $n$ is number of products. A lower Gini indicates fairer distribution.

\item \textbf{Average income gap:} The income gap is the difference between the maximum expected revenue with the minimum one across all re-sellers, i.e., 
\begin{equation}
\text{Income Gap}=\max_{\forall u_i,u_k\in U}\{\lvert rev(u_i)-rev(u_k)\rvert\}, 
\end{equation}
where $rev(u_i)=\sum_{p_j\in A_i} W_{i,j}$. 

\item \textbf{Percentage of constraint-satisfying allocations:} If we are given $X$ datasets/clusters of re-sellers and the corresponding products to be allocated, we find in what percentage of the datasets/clusters, an algorithm is able to provide an allocation satisfying \userConst and \productConst.
\end{itemize}

\subsection{Comparison of the Optimal Approaches}
The MILP formulations, i.e., \textsc{RevMax}, \textsc{ConRevMax} and \textsc{FairNash} provide optimal allocations under integer valuations. However, since MILP-based algorithms are prohibitively slow, we measure their efficacy only on the synthetic datasets.

\subsubsection{Efficacy}
Fig. ~\ref{fig:sytheticDataResult} presents the results. We observe that enforcing Nash has minimal impact on the revenue. Furthermore the fairness on re-seller side in terms of Gini is significantly better. The Gini on the product sides are similar for \textsc{RevMax} and \textsc{NashMax} since both ensure \productConst. 

\subsubsection{Scalability}
As noted earlier MILP-based algorithms are prohibitively slow, making them impractical for real datasets. We substantiate this claim by analyzing the growth of their running times against the number of re-sellers and products in the system. Figs.~\ref{fig:resellertime}-\ref{fig:producttime} present the results, where we compare \textsc{NashMax} with the greedy heuristics of \textsc{GreedyNash} and \textsc{SeAl}. As visible, the greedy heuristics are more than $2$ orders of magnitude faster.

\subsection{Efficacy of Greedy Heuristics}
The fast running times of greedy heuristics make them an attractive proposition. We therefore investigate their efficacy on the synthetic dataset against \textsc{NashMax}, which is optimal in terms of Nash.
\subsubsection{Performance on synthetic datasets}
Table~\ref{tab:CompSynthetic} presents the results. We observe that the approximation ratios of \textsc{GreedyNash} and \textsc{SeAl} are $\approx 0.99$ indicating excellent efficacy; \textsc{SeAl} performs marginally better. We also report the minimum approximation ratios across all instances of synthetic datasets and notice that it is close to the average indicating excellent stability. The dip in revenue is less than $2\%$, with \textsc{SeAl} being the better performer. Finally, the average of the income gaps is also close to the optimal. Overall, these results establish greedy heuristics as a good alternative due to obtaining excellent balance between efficiency and efficacy.

\begin{table}[t]
\centering
\scalebox{0.9} {
\begin{tabular}{lccc} 
    \hline
    {\bf } & {\bf \sealalgo} & {\bf \greedyNashalgo}\\ \hline 
    {\bf Avg. approx. ratio} & $0.9989$ & $0.997$\\ 
    {\bf Min. approx. ratio} & $0.9952$ & $0.987$\\
    {\bf \% Dip in revenue} & $0.47$ & $1.61$\\
    {\bf Avg. income gap ratio} & $0.97$ & $1.02$\\
 \hline
\end{tabular}}
\caption{Comparison of \sealalgo and \greedyNashalgo on synthetic datasets. We present the ratios when compared to the same metric obtained by \textsc{NashMax}.}
\label{tab:CompSynthetic}
\end{table}
\subsubsection{Performance on real datasets:} We next measure performance on the real datasets. Although no technique consistently achieves the best performance across all metrics, \textsc{SeAl} and \textsc{GreedyNash} achieve the best overall balance in performance. 
We notice that both \fralgo and \TFUalgo violate cardinality constraints on the product side in at least $96\%$ of the instances. Although these baseline algorithms try to adhere to two-sided constraints, they do not guarantee adherence, which we do. This property manifests itself in the trends visible in Table~\ref{tab:CompReal}. Finally, although \TFUalgo achieves the best Gini on the re-seller side, its revenue is significantly worse, making it impractical for social commerce platforms. 
\looseness=-1

\section{Related Work}
\noindent {\bf Social recommendation:}
 In a social network, users spread their preferences for items to their connections by sharing the items, and this potentially converts to the sales of items among connections. In the social commerce domain, the problem of Influence Maximization (IM) is well-studied~\cite{8_inproceedings, 9_inproceedings, 12_article, 14_article, 15_inproceedings,18_article, 23_10.1145/2723372.2723734, 24_tang2014influence}. The focal point of IM is, however, different from ours. In IM, given a social network and a budget $b$, the goal is to select $n$ nodes from the network, such that if they advertise your product, then the cascading effect of this action is maximized through network diffusion. In our problem, the set of influencers i.e., re-sellers, are given to us as input. Thus, our focus is on which influencer i.e., re-sellers, should promote which product, which is the allocation aspect. In addition, the social network between re-sellers and customers is unavailable in our case. Hence, recommendation techniques based on the rating of goods provided by users and the interaction of those users in a social network~\cite{SoRec, SocRecMatrixFac, xu2020global,long2021social,lin2019cross, xu2019relation,gao2020item} 
 is out of our scope.
 \looseness=-1

\noindent {\bf Allocation problems: }
Fair allocation has been widely studied in computational social choice theory, where the promary interest is in 
the allocation of divisible goods (aka the cake-cutting problem). Here, however, we are in the realm of indivisible goods since our allocation is binary – a product can either be allocated or not. Very recently, some papers focused on allocating indivisible goods in budgeted course allocation~\cite{6_Budish2010TheCA}, balanced graph partition~\cite{5_bouveret2017fair}, or allocation of cardinality constrained group of resources \cite{4_barman2020fair}. Few other sets of papers propose envy freeness up to one good (EF1) and Pareto optimal product allocations \cite{fairRec, fairGroupRec, aziz2021fair,ExAnteExPost}. 
Going beyond, we explore various well-studied notions of fairness in the allocation of goods ---including Nash social welfare, 
EF1, and equitability up to one item (EQ1)---in the context of two-sided cardinality constraints.\\
\noindent {\bf Fair Recommendation: }
~\citet{multiFR} proposes a multi-objective gradient descent algorithm that optimizes fairness and utility to balance the objectives between consumers and producers. Few recent works propose a framework to learn the relevance scores considering the joint multi-objective optimization \cite{ge2022toward,wu2022joint}. 
 However, in this work, we only focus on the allocation aspect. To the best of our knowledge, this is the first work on the allocation problem in the context of social commerce, which does not have a social graph among re-sellers and customers.
 \looseness=-1
\begin{table}[t]
\centering
\scalebox{0.7} {
\begin{tabular}{p{1.0in}cccccccc} 
    \toprule
    {\bf Metric } & {\bf \sealalgo} & {\bf GREEDY-} & {\bf \fralgo} & {\bf GREEDY- } & {\bf \TFUalgo} & {\bf \FOEIR}\\ 
    { } & {} & {\bf NASH} & {} & {\bf REVENUE} & {} & {}\\ \midrule
    {\bf Revenue (in USD)} & $381k$ & $372k$ & $349k$ & $\textbf{390k}$ & $167k$ & $338.9k$\\
    {\bf Avg. gini (users)} & $0.314$ & $0.273$ & $0.43$ & $0.553$ & $\textbf{0.23}$ & $0.43$\\
    {\bf Avg. income gap} & $672.11$ & $\textbf{408}$ & $1123.257$ & $1525$ & $438$ & $1338.01$ \\
    {\bf \% of Cardinality-violations} & $\textbf{0}$ & $\textbf{0}$ & $96.22$ & $\textbf{0}$ & $100$ & $0$\\
 \bottomrule
\end{tabular}}
\caption{Comparison of \sealalgo and \greedyNashalgo with baselines on real social commerce dataset. The best result for each metric (row) is highlighted in bold.}
\label{tab:CompReal}
\end{table}

\section{Conclusion}
Social commerce has emerged as a successful model to doing business. Recommendation engines form a core component of the social commerce platform and dictates the revenue earned by both producers and re-sellers. While several works have studied fairness issues on e-commerce recommendation engines, it remains unexplored in the context of social commerce. In this work, we bridged this gap. Our analysis revealed that ensuring fairness on social commerce maps to a special case of fair division of individual goods under two-sided cardinality constraints. We studied multiple fairness notions and established that many of the existential and computational guarantees of unconstrained setting do not extend to two-sided cardinality settings. Based on our discoveries, we proposed Nash social welfare the optimization function of choice. We showed that optimizing Nash under two-sided constraints is NP-hard, and hence designed polynomial-time greedy heuristics to overcome the bottleneck. Extensive experiments on real datasets obtained from a leading social commerce firm showed that the proposed heuristics provide near-optimal efficacy in a fraction of time. 
In addition, they out-performed established baseline algorithms showcasing the need to design a specialized algorithm that is cognizant of the social commerce context.
\\
\noindent {\bf Acknowledgements:} A. Gupta gratefully acknowledges the support of Google Ph.D. Fellowship. R. Vaish acknowledges support from DST INSPIRE grant no. 
DST/INSPIRE/04/2020/000107 and SERB grant no. CRG/2022/002621.
\clearpage

\bibliographystyle{ACM-Reference-Format}
\bibliography{main}


\begin{thebibliography}{55}


\ifx \showCODEN    \undefined \def \showCODEN     #1{\unskip}     \fi
\ifx \showDOI      \undefined \def \showDOI       #1{#1}\fi
\ifx \showISBNx    \undefined \def \showISBNx     #1{\unskip}     \fi
\ifx \showISBNxiii \undefined \def \showISBNxiii  #1{\unskip}     \fi
\ifx \showISSN     \undefined \def \showISSN      #1{\unskip}     \fi
\ifx \showLCCN     \undefined \def \showLCCN      #1{\unskip}     \fi
\ifx \shownote     \undefined \def \shownote      #1{#1}          \fi
\ifx \showarticletitle \undefined \def \showarticletitle #1{#1}   \fi
\ifx \showURL      \undefined \def \showURL       {\relax}        \fi
\providecommand\bibfield[2]{#2}
\providecommand\bibinfo[2]{#2}
\providecommand\natexlab[1]{#1}
\providecommand\showeprint[2][]{arXiv:#2}

\bibitem[Abdollahpouri and Burke(2019)]%
        {Multi-stake}
\bibfield{author}{\bibinfo{person}{Himan Abdollahpouri} {and} \bibinfo{person}{Robin Burke}.} \bibinfo{year}{2019}\natexlab{}.
\newblock \showarticletitle{Multi-stakeholder Recommendation and its Connection to Multi-sided Fairness}. In \bibinfo{booktitle}{\emph{RMSE 2019}}.
\newblock


\bibitem[Amanatidis et~al\mbox{.}(2022)]%
        {FairDivSurvey}
\bibfield{author}{\bibinfo{person}{Georgios Amanatidis}, \bibinfo{person}{Haris Aziz}, \bibinfo{person}{Georgios Birmpas}, \bibinfo{person}{Aris Filos-Ratsikas}, \bibinfo{person}{Bo Li}, \bibinfo{person}{Hervé Moulin}, \bibinfo{person}{Alexandros~A. Voudouris}, {and} \bibinfo{person}{Xiaowei Wu}.} \bibinfo{year}{2022}\natexlab{}.
\newblock \bibinfo{title}{Fair Division of Indivisible Goods: A Survey}.
\newblock
\newblock


\bibitem[Angelovska(2018)]%
        {forbe}
\bibfield{author}{\bibinfo{person}{Nina Angelovska}.} \bibinfo{year}{2018}\natexlab{}.
\newblock \bibinfo{title}{6 Reasons Why Europeans Don't Shop Online}.
\newblock \bibinfo{howpublished}{https://www.forbes.com/sites/ninaangelovska/2018/10/23/6-reasons-why-europeans-dont-shop-online/?sh=10fb6bb22869}.
\newblock


\bibitem[Aziz et~al\mbox{.}(2019)]%
        {aziz2021fair}
\bibfield{author}{\bibinfo{person}{Haris Aziz}, \bibinfo{person}{Ioannis Caragiannis}, \bibinfo{person}{Ayumi Igarashi}, {and} \bibinfo{person}{Toby Walsh}.} \bibinfo{year}{2019}\natexlab{}.
\newblock \showarticletitle{Fair Allocation of Indivisible Goods and Chores}. In \bibinfo{booktitle}{\emph{Proc. IJCAI}}.
\newblock


\bibitem[Barman et~al\mbox{.}(2018)]%
        {BKV}
\bibfield{author}{\bibinfo{person}{Siddharth Barman}, \bibinfo{person}{Sanath~Kumar Krishnamurthy}, {and} \bibinfo{person}{Rohit Vaish}.} \bibinfo{year}{2018}\natexlab{}.
\newblock \showarticletitle{Finding Fair and Efficient Allocations}. In \bibinfo{booktitle}{\emph{Proc. of EC}}. \bibinfo{pages}{557–574}.
\newblock


\bibitem[Bhalla(2020)]%
        {mint}
\bibfield{author}{\bibinfo{person}{Tarush Bhalla}.} \bibinfo{year}{2020}\natexlab{}.
\newblock \bibinfo{title}{Social commerce firms focus on local tongue}.
\newblock \bibinfo{howpublished}{https://www.livemint.com/technology/tech-news/social-commerce-firms-focus-on-local-tongue-11596811490216.html}.
\newblock


\bibitem[Biswas and Barman(2018)]%
        {4_barman2020fair}
\bibfield{author}{\bibinfo{person}{Arpita Biswas} {and} \bibinfo{person}{Siddharth Barman}.} \bibinfo{year}{2018}\natexlab{}.
\newblock \showarticletitle{Fair Division under Cardinality Constraints}. In \bibinfo{booktitle}{\emph{Proc. IJCAI}}.
\newblock


\bibitem[Bobadilla et~al\mbox{.}(2020)]%
        {FairCF3}
\bibfield{author}{\bibinfo{person}{Jes{\'{u}}s Bobadilla}, \bibinfo{person}{Ra{\'{u}}l Lara{-}Cabrera}, \bibinfo{person}{{\'{A}}ngel Gonz{\'{a}}lez{-}Prieto}, {and} \bibinfo{person}{Fernando Ortega}.} \bibinfo{year}{2020}\natexlab{}.
\newblock \showarticletitle{DeepFair: Deep Learning for Improving Fairness in Recommender Systems}.
\newblock \bibinfo{journal}{\emph{CoRR}}  \bibinfo{volume}{abs/2006.05255} (\bibinfo{year}{2020}).
\newblock
\showeprint[arXiv]{2006.05255}


\bibitem[Bouveret et~al\mbox{.}(2017)]%
        {5_bouveret2017fair}
\bibfield{author}{\bibinfo{person}{Sylvain Bouveret}, \bibinfo{person}{Katarína Cechlárová}, \bibinfo{person}{Edith Elkind}, \bibinfo{person}{Ayumi Igarashi}, {and} \bibinfo{person}{Dominik Peters}.} \bibinfo{year}{2017}\natexlab{}.
\newblock \bibinfo{title}{Fair Division of a Graph}.
\newblock
\newblock
\showeprint[arxiv]{1705.10239}~[cs.GT]


\bibitem[Budish(2010)]%
        {6_Budish2010TheCA}
\bibfield{author}{\bibinfo{person}{Eric Budish}.} \bibinfo{year}{2010}\natexlab{}.
\newblock \showarticletitle{The combinatorial assignment problem: approximate competitive equilibrium from equal incomes}. In \bibinfo{booktitle}{\emph{BQGT}}.
\newblock


\bibitem[Caragiannis et~al\mbox{.}(2019)]%
        {NSW}
\bibfield{author}{\bibinfo{person}{Ioannis Caragiannis}, \bibinfo{person}{David Kurokawa}, \bibinfo{person}{Hervé Moulin}, \bibinfo{person}{Ariel Procaccia}, {and} \bibinfo{person}{Junxing Wang}.} \bibinfo{year}{2019}\natexlab{}.
\newblock \showarticletitle{The Unreasonable Fairness of Maximum Nash Welfare}.
\newblock \bibinfo{journal}{\emph{ACM Transactions on Economics and Computation}}  \bibinfo{volume}{7} (\bibinfo{date}{09} \bibinfo{year}{2019}), \bibinfo{pages}{1--32}.
\newblock


\bibitem[Chakraborty et~al\mbox{.}(2017)]%
        {Chakraborty2017FairSF}
\bibfield{author}{\bibinfo{person}{Abhijnan Chakraborty}, \bibinfo{person}{Anik{\'o} Hann{\'a}k}, \bibinfo{person}{Asia~J. Biega}, {and} \bibinfo{person}{K. Gummadi}.} \bibinfo{year}{2017}\natexlab{}.
\newblock \showarticletitle{Fair Sharing for Sharing Economy Platforms}. In \bibinfo{booktitle}{\emph{Proc. FATREC}}.
\newblock


\bibitem[Chaudhury et~al\mbox{.}(2022)]%
        {productCardNSW}
\bibfield{author}{\bibinfo{person}{Bhaskar Chaudhury}, \bibinfo{person}{Yun Cheung}, \bibinfo{person}{Jugal Garg}, \bibinfo{person}{Naveen Garg}, \bibinfo{person}{Martin Hoefer}, {and} \bibinfo{person}{Kurt Mehlhorn}.} \bibinfo{year}{2022}\natexlab{}.
\newblock \showarticletitle{Fair Division of Indivisible Goods for a Class of Concave Valuations}.
\newblock \bibinfo{journal}{\emph{Journal of Artificial Intelligence Research}}  \bibinfo{volume}{74} (\bibinfo{date}{05} \bibinfo{year}{2022}), \bibinfo{pages}{111--142}.
\newblock


\bibitem[Chen et~al\mbox{.}(2020)]%
        {SocRecMatrixFac}
\bibfield{author}{\bibinfo{person}{Rui Chen}, \bibinfo{person}{Yan-Shuo Chang}, \bibinfo{person}{Qingyi Hua}, \bibinfo{person}{Quanli Gao}, \bibinfo{person}{Xiang Ji}, {and} \bibinfo{person}{Bo Wang}.} \bibinfo{year}{2020}\natexlab{}.
\newblock \showarticletitle{An enhanced social matrix factorization model for recommendation based on social networks using social interaction factors}.
\newblock \bibinfo{journal}{\emph{Multimedia Tools and Applications}}  \bibinfo{volume}{79} (\bibinfo{date}{05} \bibinfo{year}{2020}), \bibinfo{pages}{1--31}.
\newblock


\bibitem[Chen et~al\mbox{.}(2010)]%
        {8_inproceedings}
\bibfield{author}{\bibinfo{person}{Wei Chen}, \bibinfo{person}{Yifei Yuan}, {and} \bibinfo{person}{Li Zhang}.} \bibinfo{year}{2010}\natexlab{}.
\newblock \showarticletitle{Scalable Influence Maximization in Social Networks under the Linear Threshold Model}.
\newblock \bibinfo{journal}{\emph{Proc. ICDM}}, \bibinfo{pages}{88--97}.
\newblock


\bibitem[Freeman et~al\mbox{.}(2021)]%
        {twoSidedFairDivision}
\bibfield{author}{\bibinfo{person}{Rupert Freeman}, \bibinfo{person}{Evi Micha}, {and} \bibinfo{person}{Nisarg Shah}.} \bibinfo{year}{2021}\natexlab{}.
\newblock \showarticletitle{Two-Sided Matching Meets Fair Division}. In \bibinfo{booktitle}{\emph{Proc. IJCAI}}, \bibfield{editor}{\bibinfo{person}{Zhi-Hua Zhou}} (Ed.). \bibinfo{pages}{203--209}.
\newblock


\bibitem[Freeman et~al\mbox{.}(2020)]%
        {ExAnteExPost}
\bibfield{author}{\bibinfo{person}{Rupert Freeman}, \bibinfo{person}{Nisarg Shah}, {and} \bibinfo{person}{Rohit Vaish}.} \bibinfo{year}{2020}\natexlab{}.
\newblock \bibinfo{title}{Best of Both Worlds: Ex-Ante and Ex-Post Fairness in Resource Allocation}.
\newblock
\newblock


\bibitem[Freeman et~al\mbox{.}(2019)]%
        {EQ1}
\bibfield{author}{\bibinfo{person}{Rupert Freeman}, \bibinfo{person}{Sujoy Sikdar}, \bibinfo{person}{Rohit Vaish}, {and} \bibinfo{person}{Lirong Xia}.} \bibinfo{year}{2019}\natexlab{}.
\newblock \bibinfo{title}{Equitable Allocations of Indivisible Goods}.
\newblock
\newblock


\bibitem[Gao et~al\mbox{.}(2020)]%
        {gao2020item}
\bibfield{author}{\bibinfo{person}{Chen Gao}, \bibinfo{person}{Chao Huang}, \bibinfo{person}{Donghan Yu}, \bibinfo{person}{Haohao Fu}, \bibinfo{person}{Tzu-Heng Lin}, \bibinfo{person}{Depeng Jin}, {and} \bibinfo{person}{Yong Li}.} \bibinfo{year}{2020}\natexlab{}.
\newblock \showarticletitle{Item Recommendation for Word-of-Mouth Scenario in Social E-Commerce}.
\newblock \bibinfo{journal}{\emph{IEEE Transactions on Knowledge and Data Engineering}} (\bibinfo{year}{2020}).
\newblock


\bibitem[Gastwirth(1972)]%
        {gini_coef}
\bibfield{author}{\bibinfo{person}{Joseph~L Gastwirth}.} \bibinfo{year}{1972}\natexlab{}.
\newblock \showarticletitle{The estimation of the {Lorenz} curve and {Gini} index}.
\newblock \bibinfo{journal}{\emph{The review of economics and statistics}} (\bibinfo{year}{1972}), \bibinfo{pages}{306--316}.
\newblock


\bibitem[Ge et~al\mbox{.}(2022)]%
        {ge2022toward}
\bibfield{author}{\bibinfo{person}{Yingqiang Ge}, \bibinfo{person}{Xiaoting Zhao}, \bibinfo{person}{Lucia Yu}, \bibinfo{person}{Saurabh Paul}, \bibinfo{person}{Diane Hu}, \bibinfo{person}{Chu-Cheng Hsieh}, {and} \bibinfo{person}{Yongfeng Zhang}.} \bibinfo{year}{2022}\natexlab{}.
\newblock \showarticletitle{Toward Pareto Efficient Fairness-Utility Trade-off in Recommendation through Reinforcement Learning}. In \bibinfo{booktitle}{\emph{Proc. WSDM}}. \bibinfo{pages}{316--324}.
\newblock


\bibitem[Gollapudi et~al\mbox{.}(2020)]%
        {EFTwoSided}
\bibfield{author}{\bibinfo{person}{Sreenivas Gollapudi}, \bibinfo{person}{Kostas Kollias}, {and} \bibinfo{person}{Benjamin Plaut}.} \bibinfo{year}{2020}\natexlab{}.
\newblock \showarticletitle{Almost Envy-Free Repeated Matching In Two-Sided Markets}. In \bibinfo{booktitle}{\emph{WINE}} (Beijing, China). \bibinfo{publisher}{Springer-Verlag}, \bibinfo{address}{Berlin, Heidelberg}, \bibinfo{pages}{3–16}.
\newblock
\showISBNx{978-3-030-64945-6}


\bibitem[Gourv\`{e}s et~al\mbox{.}(2014)]%
        {MatroidEQProof}
\bibfield{author}{\bibinfo{person}{Laurent Gourv\`{e}s}, \bibinfo{person}{J\'{e}r\^{o}me Monnot}, {and} \bibinfo{person}{Lydia Tlilane}.} \bibinfo{year}{2014}\natexlab{}.
\newblock \showarticletitle{Near Fairness in Matroids}. In \bibinfo{booktitle}{\emph{Proc. ECAI}}. \bibinfo{pages}{393–398}.
\newblock


\bibitem[Goyal et~al\mbox{.}(2011)]%
        {9_inproceedings}
\bibfield{author}{\bibinfo{person}{Amit Goyal}, \bibinfo{person}{Wei Lu}, {and} \bibinfo{person}{Laks Lakshmanan}.} \bibinfo{year}{2011}\natexlab{}.
\newblock \showarticletitle{SIMPATH: An Efficient Algorithm for Influence Maximization under the Linear Threshold Model}.
\newblock \bibinfo{journal}{\emph{Proc. ICDM}}, \bibinfo{pages}{211--220}.
\newblock


\bibitem[Hao et~al\mbox{.}(2021)]%
        {fairCF2}
\bibfield{author}{\bibinfo{person}{Qianxiu Hao}, \bibinfo{person}{Qianqian Xu}, \bibinfo{person}{Zhiyong Yang}, {and} \bibinfo{person}{Qingming Huang}.} \bibinfo{year}{2021}\natexlab{}.
\newblock \showarticletitle{Pareto Optimality for Fairness-Constrained Collaborative Filtering}. In \bibinfo{booktitle}{\emph{Proc. ACM International Conference on Multimedia}}. \bibinfo{pages}{5619–5627}.
\newblock


\bibitem[Hartmanis(1982)]%
        {NPcomplteteProof}
\bibfield{author}{\bibinfo{person}{Juris Hartmanis}.} \bibinfo{year}{1982}\natexlab{}.
\newblock \showarticletitle{Computers and Intractability: A Guide to the Theory of NP-Completeness (Michael R. Garey and David S. Johnson)}.
\newblock \bibinfo{journal}{\emph{SIAM Rev.}} \bibinfo{volume}{24}, \bibinfo{number}{1} (\bibinfo{year}{1982}), \bibinfo{pages}{90--91}.
\newblock


\bibitem[Hu et~al\mbox{.}(2008)]%
        {cf2}
\bibfield{author}{\bibinfo{person}{Yifan Hu}, \bibinfo{person}{Yehuda Koren}, {and} \bibinfo{person}{Chris Volinsky}.} \bibinfo{year}{2008}\natexlab{}.
\newblock \showarticletitle{Collaborative Filtering for Implicit Feedback Datasets}. In \bibinfo{booktitle}{\emph{Proc. ICDM}}. \bibinfo{pages}{263--272}.
\newblock


\bibitem[Islam et~al\mbox{.}(2021)]%
        {fairCF1}
\bibfield{author}{\bibinfo{person}{Rashidul Islam}, \bibinfo{person}{Kamrun~Naher Keya}, \bibinfo{person}{Ziqian Zeng}, \bibinfo{person}{Shimei Pan}, {and} \bibinfo{person}{James Foulds}.} \bibinfo{year}{2021}\natexlab{}.
\newblock \showarticletitle{Debiasing Career Recommendations with Neural Fair Collaborative Filtering}. In \bibinfo{booktitle}{\emph{Proc. WWW}}. \bibinfo{pages}{3779–3790}.
\newblock


\bibitem[Johnson(2014)]%
        {cf1}
\bibfield{author}{\bibinfo{person}{Christopher~C. Johnson}.} \bibinfo{year}{2014}\natexlab{}.
\newblock \showarticletitle{Logistic Matrix Factorization for Implicit Feedback Data}.
\newblock


\bibitem[Jung et~al\mbox{.}(2011)]%
        {12_article}
\bibfield{author}{\bibinfo{person}{Kyomin Jung}, \bibinfo{person}{Wooram Heo}, {and} \bibinfo{person}{Wei Chen}.} \bibinfo{year}{2011}\natexlab{}.
\newblock \showarticletitle{IRIE: Scalable and Robust Influence Maximization in Social Networks}.
\newblock \bibinfo{journal}{\emph{Proc. ICDM}} (\bibinfo{date}{11} \bibinfo{year}{2011}).
\newblock


\bibitem[Kumar et~al\mbox{.}(2013)]%
        {14_article}
\bibfield{author}{\bibinfo{person}{V. Kumar}, \bibinfo{person}{Vikram Bhaskaran}, \bibinfo{person}{Rohan Mirchandani}, {and} \bibinfo{person}{Milap Shah}.} \bibinfo{year}{2013}\natexlab{}.
\newblock \showarticletitle{Practice Prize Winner---Creating a Measurable Social Media Marketing Strategy: Increasing the Value and ROI of Intangibles and Tangibles for Hokey Pokey}.
\newblock \bibinfo{journal}{\emph{Marketing Science}}  \bibinfo{volume}{32} (\bibinfo{date}{03} \bibinfo{year}{2013}), \bibinfo{pages}{194--212}.
\newblock


\bibitem[Lee(2017)]%
        {L17apx}
\bibfield{author}{\bibinfo{person}{Euiwoong Lee}.} \bibinfo{year}{2017}\natexlab{}.
\newblock \showarticletitle{{APX-Hardness of Maximizing Nash Social Welfare with Indivisible Items}}.
\newblock \bibinfo{journal}{\emph{Inform. Process. Lett.}}  \bibinfo{volume}{122} (\bibinfo{year}{2017}), \bibinfo{pages}{17--20}.
\newblock


\bibitem[Leskovec et~al\mbox{.}(2007)]%
        {15_inproceedings}
\bibfield{author}{\bibinfo{person}{Jure Leskovec}, \bibinfo{person}{Andreas Krause}, \bibinfo{person}{Carlos Guestrin}, \bibinfo{person}{Christos Faloutsos}, \bibinfo{person}{Jeanne Vanbriesen}, {and} \bibinfo{person}{Natalie Glance}.} \bibinfo{year}{2007}\natexlab{}.
\newblock \showarticletitle{Cost-effective outbreak detection in networks}.
\newblock \bibinfo{journal}{\emph{Proc. SIGKDD}}  \bibinfo{volume}{420-429}, \bibinfo{pages}{420--429}.
\newblock


\bibitem[Lesmana et~al\mbox{.}(2019)]%
        {NEURIPS2019_3070e6ad}
\bibfield{author}{\bibinfo{person}{Nixie~S Lesmana}, \bibinfo{person}{Xuan Zhang}, {and} \bibinfo{person}{Xiaohui Bei}.} \bibinfo{year}{2019}\natexlab{}.
\newblock \showarticletitle{Balancing Efficiency and Fairness in On-Demand Ridesourcing}. In \bibinfo{booktitle}{\emph{Proc. NeurIPS}}, Vol.~\bibinfo{volume}{32}.
\newblock


\bibitem[Lin et~al\mbox{.}(2019)]%
        {lin2019cross}
\bibfield{author}{\bibinfo{person}{Tzu-Heng Lin}, \bibinfo{person}{Chen Gao}, {and} \bibinfo{person}{Yong Li}.} \bibinfo{year}{2019}\natexlab{}.
\newblock \showarticletitle{Cross: Cross-platform recommendation for social e-commerce}. In \bibinfo{booktitle}{\emph{Proc. SIGIR}}. \bibinfo{pages}{515--524}.
\newblock


\bibitem[Long et~al\mbox{.}(2021)]%
        {long2021social}
\bibfield{author}{\bibinfo{person}{Xiaoling Long}, \bibinfo{person}{Chao Huang}, \bibinfo{person}{Yong Xu}, \bibinfo{person}{Huance Xu}, \bibinfo{person}{Peng Dai}, \bibinfo{person}{Lianghao Xia}, {and} \bibinfo{person}{Liefeng Bo}.} \bibinfo{year}{2021}\natexlab{}.
\newblock \showarticletitle{Social recommendation with self-supervised metagraph informax network}. In \bibinfo{booktitle}{\emph{Proc. CIKM}}. \bibinfo{pages}{1160--1169}.
\newblock


\bibitem[Ma et~al\mbox{.}(2008)]%
        {SoRec}
\bibfield{author}{\bibinfo{person}{Hao Ma}, \bibinfo{person}{Haixuan Yang}, \bibinfo{person}{Michael~R. Lyu}, {and} \bibinfo{person}{Irwin King}.} \bibinfo{year}{2008}\natexlab{}.
\newblock \showarticletitle{SoRec: Social Recommendation Using Probabilistic Matrix Factorization}. In \bibinfo{booktitle}{\emph{Proc. CIKM}}. \bibinfo{pages}{931–940}.
\newblock


\bibitem[Mansoury(2021)]%
        {FairMultiSide}
\bibfield{author}{\bibinfo{person}{Masoud Mansoury}.} \bibinfo{year}{2021}\natexlab{}.
\newblock \showarticletitle{Fairness-Aware Recommendation in Multi-Sided Platforms}. In \bibinfo{booktitle}{\emph{Proc. WSDM}}. \bibinfo{pages}{1117–1118}.
\newblock


\bibitem[Nainggolan et~al\mbox{.}(2019)]%
        {elbow}
\bibfield{author}{\bibinfo{person}{Rena Nainggolan}, \bibinfo{person}{Resianta Perangin-angin}, \bibinfo{person}{Emma Simarmata}, {and} \bibinfo{person}{Astuti~Feriani Tarigan}.} \bibinfo{year}{2019}\natexlab{}.
\newblock \showarticletitle{Improved the performance of the K-means cluster using the sum of squared error (SSE) optimized by using the Elbow method}. In \bibinfo{booktitle}{\emph{Journal of Physics: Conference Series}}, Vol.~\bibinfo{volume}{1361}.
\newblock


\bibitem[Nguyen et~al\mbox{.}(2022)]%
        {LPT}
\bibfield{author}{\bibinfo{person}{Trung~Thanh Nguyen} {et~al\mbox{.}}} \bibinfo{year}{2022}\natexlab{}.
\newblock \showarticletitle{An Experimental Study of Fast Greedy Algorithms for Fair Allocation Problems}.
\newblock \bibinfo{journal}{\emph{Journal on Information Technologies \& Communications}} \bibinfo{volume}{2022}, \bibinfo{number}{2} (\bibinfo{year}{2022}), \bibinfo{pages}{71--81}.
\newblock


\bibitem[Ohsaka et~al\mbox{.}(2014)]%
        {18_article}
\bibfield{author}{\bibinfo{person}{N. Ohsaka}, \bibinfo{person}{T. Akiba}, \bibinfo{person}{Yuichi Yoshida}, {and} \bibinfo{person}{Ken-ichi Kawarabayashi}.} \bibinfo{year}{2014}\natexlab{}.
\newblock \showarticletitle{Fast and accurate influence maximization on large networks with pruned Monte-Carlo simulations}.
\newblock \bibinfo{journal}{\emph{Proceedings of the National Conference on Artificial Intelligence}}  \bibinfo{volume}{1} (\bibinfo{date}{01} \bibinfo{year}{2014}), \bibinfo{pages}{138--144}.
\newblock


\bibitem[Patro et~al\mbox{.}(2020)]%
        {fairRec}
\bibfield{author}{\bibinfo{person}{Gourab~K Patro}, \bibinfo{person}{Arpita Biswas}, \bibinfo{person}{Niloy Ganguly}, \bibinfo{person}{Krishna~P. Gummadi}, {and} \bibinfo{person}{Abhijnan Chakraborty}.} \bibinfo{year}{2020}\natexlab{}.
\newblock \showarticletitle{FairRec: Two-Sided Fairness for Personalized Recommendations in Two-Sided Platforms}. In \bibinfo{booktitle}{\emph{Proc. WWW}} (Taipei, Taiwan) \emph{(\bibinfo{series}{WWW '20})}. \bibinfo{pages}{1194–1204}.
\newblock


\bibitem[Segal-Halevi(2019)]%
        {segal2019fair}
\bibfield{author}{\bibinfo{person}{Erel Segal-Halevi}.} \bibinfo{year}{2019}\natexlab{}.
\newblock \showarticletitle{Fair division with bounded sharing}.
\newblock \bibinfo{journal}{\emph{arXiv preprint arXiv:1912.00459}} (\bibinfo{year}{2019}).
\newblock


\bibitem[Singh and Joachims(2018)]%
        {FOEIR}
\bibfield{author}{\bibinfo{person}{Ashudeep Singh} {and} \bibinfo{person}{Thorsten Joachims}.} \bibinfo{year}{2018}\natexlab{}.
\newblock \showarticletitle{Fairness of Exposure in Rankings}.
\newblock \bibinfo{journal}{\emph{CoRR}}  \bibinfo{volume}{abs/1802.07281} (\bibinfo{year}{2018}).
\newblock
\showeprint[arXiv]{1802.07281}


\bibitem[Singh(2021)]%
        {techcruch}
\bibfield{author}{\bibinfo{person}{Manish Singh}.} \bibinfo{year}{2021}\natexlab{}.
\newblock \bibinfo{title}{Indian social commerce Meesho raises \$570 million at \$4.9 billion valuation}.
\newblock \bibinfo{howpublished}{https://techcrunch.com/2021/09/29/meesho-india-social-commerce-raises-570-million/}.
\newblock


\bibitem[strategicmarketresearch(2022)]%
        {strategicmarketresearch}
\bibfield{author}{\bibinfo{person}{strategicmarketresearch}.} \bibinfo{year}{2022}\natexlab{}.
\newblock \bibinfo{title}{Social Commerce Market Size, Share, Global Report 2030}.
\newblock \bibinfo{howpublished}{https://www.strategicmarketresearch.com/market-report/social-commerce-market}.
\newblock


\bibitem[S{\"u}hr et~al\mbox{.}(2019)]%
        {suhr2019two}
\bibfield{author}{\bibinfo{person}{Tom S{\"u}hr}, \bibinfo{person}{Asia~J Biega}, \bibinfo{person}{Meike Zehlike}, \bibinfo{person}{Krishna~P Gummadi}, {and} \bibinfo{person}{Abhijnan Chakraborty}.} \bibinfo{year}{2019}\natexlab{}.
\newblock \showarticletitle{Two-sided fairness for repeated matchings in two-sided markets: A case study of a ride-hailing platform}. In \bibinfo{booktitle}{\emph{Proc. KDD}}.
\newblock


\bibitem[Tang et~al\mbox{.}(2015)]%
        {23_10.1145/2723372.2723734}
\bibfield{author}{\bibinfo{person}{Youze Tang}, \bibinfo{person}{Yanchen Shi}, {and} \bibinfo{person}{Xiaokui Xiao}.} \bibinfo{year}{2015}\natexlab{}.
\newblock \showarticletitle{Influence Maximization in Near-Linear Time: A Martingale Approach}. In \bibinfo{booktitle}{\emph{Proc. SIGMOD}}. \bibinfo{pages}{1539–1554}.
\newblock


\bibitem[Tang et~al\mbox{.}(2014)]%
        {24_tang2014influence}
\bibfield{author}{\bibinfo{person}{Youze Tang}, \bibinfo{person}{Xiaokui Xiao}, {and} \bibinfo{person}{Yanchen Shi}.} \bibinfo{year}{2014}\natexlab{}.
\newblock \bibinfo{title}{Influence Maximization: Near-Optimal Time Complexity Meets Practical Efficiency}.
\newblock
\newblock
\showeprint[arxiv]{1404.0900}~[cs.SI]


\bibitem[Wu et~al\mbox{.}(2021b)]%
        {multiFR}
\bibfield{author}{\bibinfo{person}{Haolun Wu}, \bibinfo{person}{Chen Ma}, \bibinfo{person}{Bhaskar Mitra}, \bibinfo{person}{Fernando Diaz}, {and} \bibinfo{person}{Xue Liu}.} \bibinfo{year}{2021}\natexlab{b}.
\newblock \bibinfo{title}{Multi-FR: A Multi-Objective Optimization Method for Achieving Two-sided Fairness in E-commerce Recommendation}.
\newblock
\newblock


\bibitem[Wu et~al\mbox{.}(2022)]%
        {wu2022joint}
\bibfield{author}{\bibinfo{person}{Haolun Wu}, \bibinfo{person}{Bhaskar Mitra}, \bibinfo{person}{Chen Ma}, \bibinfo{person}{Fernando Diaz}, {and} \bibinfo{person}{Xue Liu}.} \bibinfo{year}{2022}\natexlab{}.
\newblock \showarticletitle{Joint Multisided Exposure Fairness for Recommendation}. In \bibinfo{booktitle}{\emph{Proc. SIGIR}}.
\newblock


\bibitem[Wu et~al\mbox{.}(2021a)]%
        {TFROM}
\bibfield{author}{\bibinfo{person}{Yao Wu}, \bibinfo{person}{Jian Cao}, \bibinfo{person}{Guandong Xu}, {and} \bibinfo{person}{Yudong Tan}.} \bibinfo{year}{2021}\natexlab{a}.
\newblock \showarticletitle{TFROM: A Two-Sided Fairness-Aware Recommendation Model for Both Customers and Providers}. In \bibinfo{booktitle}{\emph{Proc. SIGIR}}. \bibinfo{pages}{1013–1022}.
\newblock


\bibitem[Xiao et~al\mbox{.}(2017)]%
        {fairGroupRec}
\bibfield{author}{\bibinfo{person}{Lin Xiao}, \bibinfo{person}{Zhang Min}, \bibinfo{person}{Zhang Yongfeng}, \bibinfo{person}{Gu Zhaoquan}, \bibinfo{person}{Liu Yiqun}, {and} \bibinfo{person}{Ma Shaoping}.} \bibinfo{year}{2017}\natexlab{}.
\newblock \showarticletitle{Fairness-Aware Group Recommendation with Pareto-Efficiency}. In \bibinfo{booktitle}{\emph{Proc. RecSys}} (Como, Italy) \emph{(\bibinfo{series}{RecSys '17})}. \bibinfo{pages}{107–115}.
\newblock


\bibitem[Xu et~al\mbox{.}(2019)]%
        {xu2019relation}
\bibfield{author}{\bibinfo{person}{Fengli Xu}, \bibinfo{person}{Jianxun Lian}, \bibinfo{person}{Zhenyu Han}, \bibinfo{person}{Yong Li}, \bibinfo{person}{Yujian Xu}, {and} \bibinfo{person}{Xing Xie}.} \bibinfo{year}{2019}\natexlab{}.
\newblock \showarticletitle{Relation-aware graph convolutional networks for agent-initiated social e-commerce recommendation}. In \bibinfo{booktitle}{\emph{Proc. CIKM}}. \bibinfo{pages}{529--538}.
\newblock


\bibitem[Xu et~al\mbox{.}(2020)]%
        {xu2020global}
\bibfield{author}{\bibinfo{person}{Huance Xu}, \bibinfo{person}{Chao Huang}, \bibinfo{person}{Yong Xu}, \bibinfo{person}{Lianghao Xia}, \bibinfo{person}{Hao Xing}, {and} \bibinfo{person}{Dawei Yin}.} \bibinfo{year}{2020}\natexlab{}.
\newblock \showarticletitle{Global context enhanced social recommendation with hierarchical graph neural networks}. In \bibinfo{booktitle}{\emph{Proc. ICDM}}. IEEE, \bibinfo{pages}{701--710}.
\newblock


\end{thebibliography}
\clearpage

\appendix
\section*{Appendix}
\label{sec:paper_appendix}

\renewcommand{\thesubsection}{\Alph{subsection}}
\renewcommand{\thefigure}{\Alph{figure}}
\renewcommand{\thetable}{\Alph{table}}
\renewcommand{\thealgorithm}{\Alph{algorithm}}
\subsection{Non-Existence Proofs}
\label{sec:Non_Existence_Proof}
\theoremFirst*
\begin{proof}
Suppose there are 5 re-sellers and 5 products. We want to allocate 3 products to each re-seller and provide at least 3 copies of recommendations for each product, i.e., $L_1=L_2=R_1=R2=3$. The utility values ($W_{i,j}$) are given in Table \ref{tab:EQ1Fail}. In such a setting, these conditions must satisfy: \\(1) Out of first 4 rows, at least 1 row will have valuation $\alpha\times3/4$.\\(2) If we remove one item from allocation of $u_5$, its utility will be $4$. Now, $EQ1$ fails if $4 > \alpha\times3/4$ $\Rightarrow$ $\alpha < 16/3$.
\end{proof}
\begin{table}[!htb]
    \begin{minipage}{.50\linewidth}
      \centering
      \scalebox{0.9}{
    \begin{tabular}{llllll}
    \toprule
       & $p_1$   & $p_2$   & $p_3$   & $p_4$   & $p_5$     \\
       \midrule
    $u_1$ & $\alpha$/4 & $\alpha$/4 & $\alpha$/4 & $\alpha$/4 & 10-$\alpha$  \\
    $u_2$ & $\alpha$/4 & $\alpha$/4 & $\alpha$/4 & $\alpha$/4 & 10-$\alpha$  \\
    $u_3$ & $\alpha$/4 & $\alpha$/4 & $\alpha$/4 & $\alpha$/4 & 10-$\alpha$  \\
    $u_4$ & $\alpha$/4 & $\alpha$/4 & $\alpha$/4 & $\alpha$/4 & 10-$\alpha$  \\
    $u_5$ & 2    & 2    & 2    & 2    & 2     \\
    \bottomrule
    \end{tabular}}
    \caption{}
    \label{tab:EQ1Fail}
    \end{minipage}%
    \begin{minipage}{.50\linewidth}
      \centering
    \scalebox{0.9}{
    \begin{tabular}{llllll}
    \toprule
       & $p_1$   & $p_2$   & $p_3$   & $p_4$   & $p_5$     \\
       \midrule
    $u_1$ & 3 & 3 & 2 & 1 & 1  \\
    $u_2$ & 3 & 3 & 2 & 1 & 1  \\
    $u_3$ & 3 & 3 & 2 & 1 & 1  \\
    $u_4$ & 3 & 3 & 2 & 1 & 1  \\
    $u_5$ & 1 & 1 & 2 & 3 & 3    \\ 
    \bottomrule
    \end{tabular}}
    \caption{}
    \label{tab:EF1DependUsersPerm}
    \end{minipage} 
\end{table}

\propositionThree*
\begin{proof}
Suppose there are 5 re-sellers and 5 products. We want to allocate 3 products to each re-seller and provide at least 3 copies of recommendation to each product. The utility values ($W_{i,j}$) are given in Table \ref{tab:EF1DependUsersPerm}. In such a setting, we can not find a feasible $EF1$ solution if we take any permutation that does not start with $u_5$.
\end{proof}

\theoremThree*
\begin{proof}
Consider an instance with two re-sellers $u_1$ and $u_2$ and four products $p_1,\: p_2,\: p_3$ and $p_4$. Suppose each re-seller can receive exactly two products, and suppose each product can be assigned to exactly one re-seller, i,e, $L_1=L_2=2$ in \userConst and $R_1=R_2=1$ in \productConst.

Let the utilities for re-seller $u_1$ be $1,\:1,\: 2+\epsilon,\: 2+\epsilon$ for $p_1,\:p_2,\:p_3,\:p_4$ respectively, where $\epsilon>0$ is sufficiently small (say less than $0.4$). Similarly, let re-seller $u_2$ value $p_1,\: p_2,\; p_3,\; p_4$ at $\epsilon,\;\epsilon,\; 3,\; 3$, respectively.
A cardinality-constrained Nash optimal allocation gives {$p_1,p_2$} to re-seller $u_1$ and {$p_3,\:p_4$} to re-seller $u_2$. This violates $EF1$ w.r.t. re-seller $u_1$. 
\end{proof}

\subsection{Hardness Proofs}
\label{sec:Hardness_Proof}
\theoremSecond*
\begin{proof}
We will show a reduction from the \textit{3-Partition} problem, which is known to be NP-complete~\cite{NPcomplteteProof}. 

\begin{definition}[3-Partition problem] An instance of the 3-Partition problem involves a set of 3$m$ numbers $b_1,\cdots,b_{3m}$ that sum up to $mT$. The goal is to partition the numbers into $m$ triplets with equal sums. The problem remains NP-complete even if all $b_i$'s are close to $T/3$.
\looseness=-1
\end{definition}

\noindent
\textbf{Fair division instance to find \boldmath$EQ1$ under cardinality constraints:} Construct a fair division instance with a set of $m+1$ re-sellers and $3m+3$ items. There are $m$ main re-sellers and one dummy re-seller. The items consist of $3m$ partition items $g_1,\cdots,g_{3m}$ and three dummy items. Each main re-seller values each partition item $g_i$ at $b_i$, and each dummy item at 0. The dummy re-seller values every item at $T/2$. The cardinality constraints require that each re-seller receives exactly three items, and one copy of the items is allocated.\\
\textbf{In the forward direction}, given a solution to 3-Partition, one can construct a feasible $EQ1$ allocation by assigning the partition items among the main re-sellers in accordance with the solution to 3-Partition, and give the three dummy items to the dummy re-seller.\\
\textbf{In the reverse direction}, suppose $A$ is a feasible $EQ1$ allocation. Then, by the cardinality constraint, the dummy re-seller must receive three items under $A$, each of which it values at $T/2$. Therefore, to satisfy $EQ1$, every main re-seller must have a value of at least $T$. This means that each main re-seller must get three partition items (recall that each $b_i$ is close to $T/3$, so any two $b_i$'s together cannot take us to $T$). This, in turn, implies that all dummy items must go to the dummy re-seller. Therefore, under $A$, we must allocate partition items so that each main re-seller gets exactly three partition items, and its utility is at least $T$. Since the partition items have a total value of $mT$, each main re-seller must have a value of exactly $T$, giving us a solution of 3-Partition. Hence, finding $EQ1$ allocation under cardinality constraints is NP-complete.
\end{proof}

\theoremFour*
\begin{proof} (sketch)
We briefly outline a reduction from a balanced version of the Partition problem~\cite{segal2019fair}. The input consists of a set of $2r$ numbers $a_1,\dots,a_{2r}$ that sum up to $2T$. The goal is to partition these numbers into two sets of cardinality $r$ each such that the numbers in each set sum to $T$.

The reduced fair division instance is constructed as follows: There are $2r$ items/products $p_1,\dots,p_{2r}$. Each item can be allocated to exactly one reseller (i.e., $R1 = R2 = 1$). Each reseller can get exactly $r$ items (i.e., $L1 = L2 = r$). The decision threshold for Nash welfare is $T$.

The implication in forward direction is easy to see. In the reverse direction, we observe that the arithmetic mean of the value is exactly $T$ and the geometric mean is at least $T$. Thus, by the AM-GM inquality, the utilities of resellers must be equal, inducing the desired partition.
%
\looseness=-1
\end{proof}

\subsection{Complexity Analysis}
\label{sec:complexity}
Suppose there are $m$ re-sellers and $n$ products. The minimum and maximum size of allocation list for each re-seller is $L_1$ and $L_2$ respectively. 
$R_1$ and $R_2$ are the minimum and maximum number of assigned re-sellers for each product respectively.
\subsubsection{The complexity of \greedyNashalgo algorithm}
The complexity of algorithm ~\ref{algo:algoIGNSW2} (\greedyNashalgo) is $O(mn+{R_1}^2{L_2}^3)$. The complexity of steps (1-9 and 15-18) is $O(mn)$. The greedy replacement procedure (steps 10-14) is bounded with $O({R_1}^2{L_2}^3)$ for $R_1 = \floor{\alpha \times L_1 \times m /n}$, where $0<\alpha\leq1$. At step 7, a product ($p$) is not allocated to any re-seller, when at most $R_1$ re-sellers are assigned products less than $L_2$. Note that these re-sellers are already in allocation of ($p$), say such re-sellers $U_R$. This case of not finding feasible re-seller can arise for at most $L_2$ products, say those products $P_L$. The re-sellers $U_R$ can be assigned with at most ${R_1}\times L_2$ unique products whose cardinality can exceed $R_1$ (at step 9). So there can be at most ${R_1}\times L_2 \times L_2$ re-sellers which are assigned products whose cardinality exceed $R_1$, say such re-sellers $U_{R{L_2}^2}$. To balance the cardinality of products $P_L$, at most ${R_1}\times L_2$ assignments are done by greedily replacing ${R_1}\times L_2$ products across ${R_1}\times L_2\times L_2$ re-sellers of $U_{R{L_2}^2}$. So complexity of replacement procedure (steps 10-14) is bounded with $O({R_1}^2{L_2}^3)$.\\
\subsubsection{The complexity of \sealalgo algorithm}
The complexity of algorithm ~\ref{algo:algoSEAL} (\sealalgo) is $O(mn\log m+{L_2}^2{R_1}^2)$. The complexity of steps (1-4) is $O(m\log m)$, steps (5-9 and 17-20) is $(mn\log m)$, and steps (10-11 and 21-22) is $O(m\log m)$. The complexity of steps (12-16) is bounded by $O({L_2}^2{R_1}^2)$ for ${R_1} = \floor{\alpha \times L_2 \times m /n}$, where $0<\alpha\leq1$. A re-seller can not find feasible products at step 8, if there are at most $L_2$ products which are allocated to less than $R_1$ re-sellers. Say such products $P_L$. There can be at most $L_2 \times {R_1}$ re-sellers which are allocated to products whose cardinality exceed ${R_1}$, say such re-sellers $U_{LR}$. To balance the cardinality of products $P_L$, at most $R_1 \times L_2$ assignments are done by greedily replacing ${R_1} \times L_2$ products across $R_1\times L_2$ re-sellers of $U_{LR}$. So complexity of steps (12-16) is $O({L_2}^2{R_1}^2)$.
\looseness=-1

\subsection{Comparing \greedyNashalgo and \sealalgo}
\label{sec:Examples}

\begin{table}[!htb]
    \begin{minipage}{.45\linewidth}
      \centering
      \scalebox{0.8} {
        \begin{tabular}{llll}
        \toprule
           & $p_1$   & $p_2$   & $p_3$ \\
           \midrule
        $u_1$ & 7 & 1 & 2   \\
        $u_2$ & 5.5 & 2 & 2.5  \\
        $u_3$ & 5 & 4 & 1   \\
        \bottomrule
        \end{tabular}}
    \caption{\greedyNashalgo\\ is better than \sealalgo}
    \label{tab:GNbetter}
    \end{minipage}%
    \begin{minipage}{.45\linewidth}
      \centering
      \scalebox{0.8} {
        \begin{tabular}{llll}
        \toprule
           & $p_1$   & $p_2$   & $p_3$ \\
           \midrule
        $u_1$ & 7 & 1 & 2   \\
        $u_2$ & 6 & 1.5 & 2.5  \\
        $u_3$ & 5 & 4 & 1   \\
        \bottomrule
        \end{tabular}}
    \caption{\sealalgo is better than \greedyNashalgo}
    \label{tab:SEALbetter}
    \end{minipage} 
\end{table}

While comparing \sealalgo and \greedyNashalgo, either one might output allocations with better Nash Welfare. Lets demontrate it with an example. Suppose there are three re-sellers and three products. We want to allocate 2 products to each re-seller and provide 2 copies of recommendations for each product, i.e., $L_1=L_2=R_1=R_2=2$. The utility values ($W_{i,j}$) are given in Tables \ref{tab:GNbetter} and \ref{tab:SEALbetter}. Suppose we allocate the first product to users in order $u_1$, $u_2$, and $u_3$, and we choose product ordering $p1$, $p_2$, and $p_3$ for allocation. In case of Table \ref{tab:GNbetter} \greedyNashalgo provide allocations ($u_1[p_1,p_3]$, $u_2[p_1,p_2]$,$u_3[p_2,p_3]$) with $Nash=337.5$, while \sealalgo provides allocations ($u_1[p_1,p_2]$, $u_2[p_1,p_3]$,$u_3[p_2,p_3]$) with $Nash=320$. In cases of Table \ref{tab:GNbetter}, \greedyNashalgo provides allocations ($u_1[p_1,p_3]$, $u_2[p_1,p_2]$,$u_3[p_2,p_3]$) with $Nash=337.5$, while \sealalgo provides allocations ($u_1[p_1,p_2]$, $u_2[p_1,p_3]$,$u_3[p_2,p_3]$) with $Nash=340$. 

\subsection{Importance of \userConst (Eq.\ref{eq:user_constrain})}
\label{sec:ImpUserConst}
\userConst constraint is a practical necessity. Social-commerce platforms work within the design constraints of a UI, which typically shows a certain number of products. length of the recommendation list for each re-seller is controlled through re-seller-side constraints. To verify this empirically, we benchmark our algorithms against two adapted baselines without \userConst.
\looseness=-1
\begin{itemize}
\item \textsc{LPT}~\cite{LPT} is an extension of Graham(1966)'s classical "Longest-processing-time-first(LPT)" algorithm for job scheduling. We adapt it to showcase the use of re-seller side cardinality constraints while optimizing Nash. In our adaptation, we iteratively assign products to the least-happy reseller while ensuring that product-side cardinality constraints are satisfied.
\item \UnconsGreedyNash is un-constrained version of \greedyNashalgo where \userConst (Eq.\ref{eq:user_constrain}) is not enforced.
\looseness=-1
\end{itemize}
We set $R_2=R_1+1$ for comparing \LPT with \greedyNashalgo and \sealalgo.
\begin{table}[t]
\centering
\scalebox{0.7} {
\begin{tabular}{cccccccc} 
    {\bf Metric } & {\bf \sealalgo} & {\bf \greedyNashalgo} & {\bf \LPT}& {\bf \UnconsGreedyNash}\\ \midrule
    {\bf Revenue (in USD)} & $\textbf{310.57k}$ & $309.4k$ & $306.35k$ & $308.5k$ \\
    {\bf Avg. gini (users)} & $0.19$ & $0.21$ & $\textbf{0.07}$ & $0.09$ \\
    {\bf Avg. income gap} & $398.33$ & $1080.40$ & $\textbf{342.94}$ & $346.68$  \\
    {\bf 5th $\%$tile of products allocated} & $\textbf{15}$ & $\textbf{15}$ & 1 & $1$\\
    {\bf 95th $\%$tile of products allocated} & $\textbf{15}$ & $\textbf{15}$ & 19 & $26$\\
    {\bf Variance in products allocated} & $\textbf{0}$ & $\textbf{0}$ & $30.95$& $53.81$  \\
    
\end{tabular}}
\caption{Comparison of \sealalgo, \greedyNashalgo, \LPT, and \UnconsGreedyNash on real social commerce dataset to show the need of re-seller side constraints}
\label{tab:CompLPT_A}
\end{table}
\sealalgo and \greedyNashalgo outperforms \LPT and \UnconsGreedyNash (see Table~\ref{tab:CompLPT_A}). Specifically, when \UnconsGreedyNash optimizes Nash, the number of products allocated to a reseller varies widely, violating re-seller side constraints. For example, 15.35$\%$ resellers got less than 5 products, whereas 25.82$\%$ resellers got more than 20 products. A Similar trend is observed for the \LPT too. Since selling success depends on multiple factors such as catalog diversity, finding interested customers, etc., such variability hurts the re-seller experience.
\subsection{Synthetic Dataset Generation}
\label{sec:SyntheticDataGeneration}
To compare the proposed heuristic to the optimal allocation, we generate $100$ synthetic datasets with each containing $100$ users and $100$ products. Each product's (revenue) is set to an integer chosen uniformly at random in the range $[1,1000]$. Note that Nash optimization is free from scaling the utilities of users, and the need for integer valuation for MILP is mentioned in \S~\ref{sec:ILP_Nash}. Each dimension (product) in the expertise vector of an user is set to a value chosen uniformly at random from the range $[0,1]$.
\looseness=-1

\end{document}